\documentclass[aps,pra,letterpaper,superscriptaddress] {revtex4}

\usepackage{fullpage} 
\usepackage{bbm}
\usepackage{graphicx}
\usepackage{upref}
\usepackage{enumerate}
\usepackage{latexsym}
\usepackage{color,graphics}
\usepackage{comment} 
\usepackage{relsize}
\usepackage{hyperref} 
\usepackage[redeflists]{IEEEtrantools}

\usepackage[section]{algorithm} 
\usepackage{amsmath,amsthm,amssymb,algorithmic,epsfig,graphicx, tikz}

\usepackage{amsfonts}
\usepackage{varioref}
\usepackage[ansinew]{inputenc}

\usepackage{amsmath}
\usepackage{amsfonts}
\usepackage{tikz} 

\usepackage{amssymb}
\usepackage{varioref}
\usepackage[ansinew]{inputenc}

\newtheorem{theorem}{Theorem}[section]
\newtheorem{lemma}[theorem]{Lemma}
\newtheorem{claim}[theorem]{Claim}
\newtheorem{proposition}[theorem]{Proposition}

\newtheorem{defn}[theorem]{Definition}

\newtheorem{result}{Result}
\newtheorem{definition}{Definition}

\newcommand{\R}{\mathbb{R}}

\def\ket#1{\mathinner{|{#1}\rangle}}

\renewcommand{\part}[2]{\frac{\partial #1}{\partial #2}}

\newcommand{\als}[1]{\begin{align*}#1\end{align*}}
\newcommand{\all}[2]{\begin{align}\label{#2} #1\end{align}}
\newcommand{\al}[1]{\begin{align} #1\end{align}}

\newcommand{\en}[1]{\left ( #1 \right )}

\newcommand{\nl}{\notag \\}

\newcommand{\norm}[1]{\lVert#1\rVert}

\newcommand{\thmref}[1]{\hyperref[#1]{{Theorem~\ref*{#1}}}}
\newcommand{\lemref}[1]{\hyperref[#1]{{Lemma~\ref*{#1}}}}
\newcommand{\remref}[1]{\hyperref[#1]{{Remark~\ref*{#1}}}}
\newcommand{\corref}[1]{\hyperref[#1]{{Corollary~\ref*{#1}}}}
\newcommand{\eqnref}[1]{\hyperref[#1]{{Equation~(\ref*{#1})}}}
\newcommand{\claimref}[1]{\hyperref[#1]{{Claim~\ref*{#1}}}}
\newcommand{\remarkref}[1]{\hyperref[#1]{{Remark~\ref*{#1}}}}
\newcommand{\propref}[1]{\hyperref[#1]{{Proposition~\ref*{#1}}}}
\newcommand{\factref}[1]{\hyperref[#1]{{Fact~\ref*{#1}}}}
\newcommand{\defref}[1]{\hyperref[#1]{{Definition~\ref*{#1}}}}
\newcommand{\exampleref}[1]{\hyperref[#1]{{Example~\ref*{#1}}}}
\newcommand{\hypref}[1]{\hyperref[#1]{{Hypothesis~\ref*{#1}}}}
\newcommand{\secref}[1]{\hyperref[#1]{{Section~\ref*{#1}}}}
\newcommand{\chapref}[1]{\hyperref[#1]{{Chapter~\ref*{#1}}}}
\newcommand{\apref}[1]{\hyperref[#1]{{Appendix~\ref*{#1}}}}

\AtBeginDocument{}

\begin{document}
\bstctlcite{IEEEexample:BSTcontrol}
\title{Quantum gradient descent for linear systems and least squares}  

\author{Iordanis Kerenidis} 

\thanks { Email: jkeren@irif.fr.} 
\affiliation{CNRS, IRIF, Universit\'e Paris Diderot, Paris, France} 
\affiliation{ Centre for Quantum Technologies, National University of Singapore, Singapore. }

\author{Anupam Prakash} 
\thanks { Email: anupamprakash1@gmail.com.} 
\affiliation{CNRS, IRIF, Universit\'e Paris Diderot, Paris, France.}

\begin{abstract} 
Quantum machine learning and optimization are exciting new areas that have been brought forward by the breakthrough quantum algorithm of Harrow, Hassidim
and Lloyd for solving systems of linear equations. The utility of {classical} linear system solvers extends beyond linear algebra as they can be leveraged to solve optimization problems using iterative methods like gradient descent. In this work, we provide the first quantum method for performing gradient descent when the gradient is an affine function. Performing $\tau$ steps of the gradient descent requires time $O(\tau C_S)$ for weighted least squares problems, where $C_S$ is the cost of performing one step of the gradient descent quantumly, which at times can be considerably smaller than the classical cost. We illustrate our method by providing two applications: first, for solving positive semidefinite linear systems, and, second, for performing stochastic gradient descent for the weighted least squares problem with reduced quantum memory requirements. We also provide a quantum linear system solver in the QRAM data structure model that provides significant savings in cost for large families of matrices.
\end{abstract} 
\thispagestyle{empty}

\maketitle

\newpage

\setcounter{page}{1}

\section{Introduction} 

Quantum Machine Learning has seen a flurry of recent exciting developments, mainly due to the breakthrough algorithm of Harrow, Hassidim, and Lloyd \cite{HHL09}, that takes as input a sparse and well-conditioned system of linear equations, and in time polylogarithmic in the system's dimension outputs the solution vector as a quantum state. In other words, given a matrix $A$ and a vector $b$, it outputs the quantum state $\ket{A^{-1} b}$ corresponding to the solution. 

The HHL algorithm has been very influential, and several works have obtained quantum algorithms for machine learning problems under similar assumptions \cite{LMR14, LMR14a, LMR13a}. However, the applications of the HHL algorithm to machine learning had several caveats as pointed out in \cite{A15}. Namely, the algorithm produces a quantum state $\ket{A^{-1} b}$ instead of a classical solution and achieves an exponential speedup only when the matrix is well conditioned and has sparsity poly-logarithmic in the dimension. In machine learning settings, the input matrices are not expected to be sparse or well structured, thus restricting the applications of the HHL algorithm. 

We introduced the QRAM data structure model \cite{KP16} in order to overcome some of these caveats. We presented a quantum algorithm for competitive recommendation systems in \cite{KP16} that provides a good recommendation to most users in expected time $O(\frac{\sqrt{k}}{\epsilon} \text{polylog}(mn))$, which is poly-logarithmic in matrix dimensions and 
has a square root dependence on the rank $k$ which is much smaller than the matrix dimensions.  The best known classical recommendation systems at the time required time polynomial in the matrix dimensions for the same task. 

Very recently, a classical recommendation system with running time $O(\text{poly}(k, \frac{1}{\epsilon} ) \text{polylog}(mn))$ was given \cite{T18} using classical instantiations of the QRAM data structures. We note that the classical recommendation algorithm \cite{T18} incurs a large polynomial overhead over the quantum algorithm and has a running time $O(\frac{k^{12}}{\epsilon^{12}} \text{polylog}(mn))$, which for inputs of interest make the new classical algorithm slower even than older classical algorithms. It remains an open question to obtain classical algorithms with a smaller polynomial overhead in $k$ and $1/\epsilon$. 

In this work, we generalize the quantum singular value estimation algorithm using QRAM data structures \cite{KP16} and obtain speedups for classes of dense matrices beyond the sparse and the low-rank cases. 
Further, we develop a framework for quantum gradient descent with affine updates with cost polynomial in the number of steps and present applications to weighted least squares problems. 

Our work is motivated by the applications of classical linear system solvers to optimization and machine learning using iterative methods. 
Examples of iterative methods include first-order methods like gradient descent and second order methods like the interior point method for linear and semidefinite programs. These methods allow us to solve a variety of problems that do not have a closed form solution and thus cannot be solved using linear systems alone. 

The iterative methods start with some initial solution state $\theta_0$ and update it iteratively according to a rule of the form $\theta_{t+1} = \theta_t + \alpha r_t$. In many cases, the updates are implemented using linear algebra operations like matrix multiplication and inversion. 
This motivates the question of whether quantum linear system solvers can be similarly used to solve more general optimization problems through iterative methods. 

However, there are obvious obstacles towards realizing a quantum iterative method. Quantum routines for matrix multiplication and inversion output a quantum state and not a classical vector that can be used directly for the next step of the method. 
Further, they succeed only with a fairly small probability and destroy the input when they fail. Hence, if the quantum procedure for finding the updates fails at any step, one needs to restart the iterative method from the very beginning. Algorithmically this implies that the running time of a naive quantum iterative method is exponential in the number of steps. Another problem is that the HHL based quantum linear systems solvers provide exponential speedups only for matrices that are poly-logarithmically sparse in the dimension.  

In this work, we make progress on both the challenges described above. First, we provide a first order quantum iterative method for performing gradient descent with affine update rules in time polynomial in the number of steps. Secondly, we design a quantum linear systems solver in the QRAM data structure model that allows us to obtain potentially significant speedups for many classes of dense matrices. We provide as applications iterative quantum algorithms for the weighted least squares problem including one that can perform stochastic gradient descent. 

The weighted least squares problem is widely used for regression or data fitting. The data matrix for least squares is often very skewed, i.e., the number of data points is vastly larger than the dimension of the points. Stochastic gradient descent is advantageous in such settings as it estimates the gradient by computations on a small sized batch of data points and avoids costly computations on large matrices. The quantum stochastic gradient descent algorithm allows us to reduce the size of the required QRAM as at any given time only a small subset of the data needs to be in quantum memory. 

Quantum algorithms for gradient estimation have been studied before \cite{J05, B05, J09, GAW17}, where a quantum algorithm is used to estimate the gradient of a given function, while a classical algorithm performs the gradient descent. A quantum gradient descent algorithm for polynomial optimization problems has also been presented \cite{RSPS16}, but the running time of the algorithm depends exponentially on the number of steps. Our algorithm achieves a polynomial scaling instead in the number of steps when the gradient is an affine function, a case that frequently arises in practice for solving least squares problems. 

\subsection{The QRAM data structure model} 

The QRAM data structure model requires a full-scale quantum computer equipped with a quantum random access memory.  
The QRAM is analogous to a classical random access memory and the model can be viewed as a quantum analog of using 
data structures to speed up computation. Let us first define the QRAM more precisely and compare it to the more standard oracle model. 
\begin{definition} 
A quantum random access memory (QRAM) is a device that stores data $(i,x_{i}) \in [n]$ and allows queries of the form $\ket{i, 0} \to \ket{i, x_{i}}$ to be made in superposition
in time $O($\normalfont{polylog}$(n))$.
\end{definition} 
\noindent The standard oracle model also assumes that one can make the query $\ket{i, 0} \to \ket{i, x_{i}}$ in poly-logarithmic time. However,  in the absence of a QRAM, queries $\ket{i, 0} \to \ket{i, x_{i}}$ can be made efficiently only if $x_{i}$ is a function of $i$ that can be computed using a quantum circuit of size poly-logarithmic in $N$. The QRAM implements the standard oracle for arbitrary data.

In addition, we require a notion of efficient QRAM data structures. A QRAM data structure is efficient if it avoids additional time and space overheads over storing the data sequentially in the QRAM and can be updated efficiently.  
 \begin{definition} \label{d2}
A QRAM data structure storing $n$ entries is said to be efficient it its memory requirement is $\widetilde{O}(n)$ and the cost of updating, inserting or deleting a single entry is $O($\normalfont{polylog}$(n))$.\end{definition} 
\noindent Let us elaborate this notion a bit further. A dataset of size $n$ can be stored in the QRAM by storing each of its entries sequentially in a simple list or array. The total memory used is $\widetilde{O}(n)$ (with logarithmic factors to account for memory management) and the update time for a single entry $O(\log(n))$. An efficient data structure allows us to store the dataset in a more elaborate form, for example in a binary tree, with the same memory and update time requirements as for when the dataset is stored sequentially.

 We are now ready to define the QRAM data structure model used for our algorithms,
 \begin{definition} \label{d3}
An algorithm in the QRAM data structure model that processes a data-set of size $m$ has two steps: 
\begin{enumerate} 
\item A pre-processing step with complexity $\widetilde{O}(m)$ that constructs efficient QRAM data structures for storing the data. 
\item A computational step where the quantum algorithm has access to the QRAM data structures constructed in step 1. 
\end{enumerate} 
The complexity of algorithms in this model is measured by the cost for step 2.
\end{definition} 
\noindent In order to see the relevance of this model, consider a data-processing setting where it is feasible to store a dataset of size $n$, but infeasible to perform a super-linear computation with cost $O(n^2)$ or $O(n^3)$.  In this setting, a quantum algorithm in the QRAM data structure model requires space $\widetilde{O}(n)$ and could perform the same intensive computation in time sub-linear in $n$, thus making the application feasible. It is therefore an effective model for data-processing applications. 

As our algorithms require a QRAM, we briefly discuss the feasibility of implementing a QRAM.  A bucket brigade architecture for a QRAM has been proposed \cite{GLM08a}, where only $O(\log n)$ gates from a circuit of size $O(n)$ are active when a query is made. An error analysis for this architecture was given \cite{AGJMP15}, where it was shown that the error per gate for applications like the HHL algorithm and other quantum machine learning algorithms needs to be inverse in the number of active gates, i.e., inverse in the logarithm of the number of qubits, and quantum error correction may not be required for such applications. 

We are now ready to present our main results and discuss the ideas used for the proofs. We present our results on linear system solvers in Section \ref{one2}, the quantum gradient iterative method in Section \ref{one3} and applications of the quantum iterative method in Section \ref{one4}.

\subsection{An improved quantum linear systems solver} \label{one2}

We first introduce some quantum and linear algebra notation that helpful for the discussion of our main results and is used throughout the paper. 
The set $\{ 1, 2, \cdots, n\}$ is denoted by $[n]$, the standard basis vectors in $\R^{n}$ are denoted by $e_{i}, i \in [n]$.
 For a vector $x \in \R^n$ we denote the $\ell_{p}$-norm as $\norm{x}_{p} = (\sum_{i} x_i^p)^{1/p}$. The Euclidean norm $\norm{x}_{2}$ is denoted as $\norm{x}$. The rank of a matrix is denoted as $rk(A)$. A matrix is positive semidefinite (psd) if it is symmetric and has non-negative eigenvalues, the notation $A \succeq 0$ indicates that $A$ is a psd matrix. The singular value decomposition of a symmetric matrix $A\in \R^{n\times n}$ is written as $A= \sum_{i} \lambda_{i} v_{i} v_i^T$ where $\lambda_{i} \geq 0$ are the eigenvalues and $v_{i}$ are the corresponding eigenvectors.
 
The singular value decomposition of $A\in \R^{m\times n}$ is written as $A= \sum_{i} \sigma_{i} u_{i} v_{i}^{T} $ where $\sigma_{i}$ are the singular 
values and $u_{i}$ and $v_{i}$ are the left and right singular vectors. The Frobenius norm is defined as $\norm{A}_{F}^{2}: = \sum_{ij} A_{ij}^{2} = \sum_{i} \sigma_{i}^{2}$ while the spectral norm $\norm{A} = \sigma_{max}$, where $\sigma_{\max}$ is the largest singular value. The condition number $\kappa(A) = \sigma_{max}/\sigma_{min}$. 

The $i$-th row of matrix $A \in \R^{m\times n}$ is denoted as $a_{i}$ and the $j$-th column is denoted as $a^{j}$. The $\circ$ operator denotes the Hadamard product, that is $A = P \circ Q$ implies that $A_{ij} = P_{ij}. Q_{ij}$ for $i \in [m], j \in [n]$. For a matrix $A \in \R^{m \times n}$, the maximum of the $p$-th power of the $\ell_{p}$ norm for the row vectors is denoted $s_{p}(A) := \max_{i \in [m]} \norm{a_{i}}_{p}^{p}$, the maximum of the $p$-th power of the $\ell_{p}$ norm of the column vectors is $s_{p}(A^{T} )$. The sparsity $s(A)$ is the maximum number of non-zero entries in a row of $A$. The $\widetilde{O}$ notation is used to suppress factors poly-logarithmic in vector or matrix dimensions, that is $O(f(n) \text{polylog} (mn))$
is written as $\widetilde{O}( f(n) )$. 

A vector $v \in \R^{n}$ is encoded into $O(\log n)$ dimensional normalized vector state $\ket{v} = \frac{1}{\norm{v}} \sum_{i \in [n]} v_{i} \ket{i}$. 
However, there is one exception to this use of notation, in the gradient descent procedure we denote the unnormalized garbage state by $\ket{G}$. In all other cases, the notation $\ket{v}$ where $v \in \R^{n}$ denotes the normalized vector state.

We first present our results on linear system solvers in the QRAM data structure model. The linear system solver that we present here can offer considerable speedups over classical algorithms for the problem sampling from the solutions of linear systems for large families of dense matrices, including matrices whose rows have bounded $\ell_{1}$ norm. 

The main technical tool for our linear system solver is algorithm for singular value estimation (SVE) that generalizes the algorithm in \cite{KP16}. 
More precisely, let $A= \sum_{i \in [k]} \sigma_{i} u_{i} v_{i}^{T} $ be the singular value decomposition for $A \in \R^{m\times n}$ 
where $k = \min(m,n)$ and $\sigma_{i}$ are the (possibly $0$) singular values and $u_{i}$ and $v_{i}$ are the left and right singular vectors. A given vector $b$ can be viewed as a superposition of the singular vectors of the matrix $A$, i.e. $\ket{b} = \sum_i \beta_i \ket{v_i}$. 

\begin{definition} \label{sve} 
Let $A= \sum_{i \in [k]} \sigma_{i} u_{i} v_{i}^{T} $ be the singular value decomposition for $A \in \R^{m \times n}$ for $k=\min(m,n)$ and let $\delta>0$. The singular value estimation (SVE) problem with error $\delta$ is defined as: Given $\ket{b}= \sum_{i\in [k]} \beta_i \ket{v_i}$, to perform the mapping  
\[
 \sum_{i\in [k]}  \beta_i \ket{v_i} \ket{0} \rightarrow \sum_{i \in [k]}  \beta_i \ket{v_i}\ket{\overline{\sigma}_i},
\]
such that $|\overline{\sigma}_i - \sigma_i| \leq \delta$ for all $i \in [k]$ with probability $1-1/poly(n)$. 
\end{definition} 
\noindent We had provided an SVE algorithm with running time $\widetilde{O}( \norm{A}_{F}/\delta)$ in our work on recommendation systems \cite{KP16}. This algorithm 
relied on a particular efficient QRAM data structure for storing $A$. In this work, we provide efficient QRAM data structures for storing $A$ that allow us to obtain an SVE algorithm 
with a running time that depends on the maximum $\ell_{1}$-norm of the rows of $A$, instead of the $\norm{A}_{F}$. 

Our SVE algorithms are based on the relation between quantum walks and singular values \cite{S04, C10}. We show that for every decomposition $A/\mu= P \circ Q$ such the norms of the rows of $P$ and the columns of $Q$ are at most $1$ and $\circ$ denotes the entrywise (Hadamard) product, if we have access to unitaries that prepare normalized quantum states corresponding to the rows and columns of $P, Q$ in time $\widetilde{O}(1)$, there is an SVE algorithm with running time $\widetilde{O}(\mu/\delta)$.  

\begin{result}\label{one} 
[Theorem \ref{sveplus}] 
Let $A \in \R^{m\times n}$ be a matrix and suppose there exist $P, Q \in \R^{m\times n}$ and $\mu >0$ such that
$\norm{p_{i}}_{2} \leq 1 \; \forall i \in [m], \;\norm{q^{j}}_{2} \leq 1 \; \forall j \in [n]$ and
\all{ 
A/\mu = P \circ Q.  
} {hadamard} 
If unitaries $U: \ket{i} \ket{0^{\lceil \log (n+1) \rceil}} \to \ket{i} \ket{\overline{p}_{i}}$ and $V:  \ket{0^{\lceil \log (m+1) \rceil}}  \ket{j} \to \ket{\overline{q}^{j}} \ket{j}$ can be implemented in time $\widetilde{O} (\log (mn))$ then there is an algorithm that performs SVE for 
$A$ in time $\widetilde{O}(\mu/\delta)$. 
\end{result}

This result also explains the connection between a QRAM data structure and an SVE algorithm. The QRAM data structure is designed with a particular factorization $A/\mu = P \circ Q$ in mind, it provides an $\widetilde{O}(1)$ time implementation of the unitaries $U$ and $V$ and therefore an SVE algorithm with running time $\widetilde{O}(\mu/\delta)$ by Result \ref{one}. 

The SVE algorithm in \cite{KP16} falls in this framework, it utilizes an efficient QRAM data structure corresponding to a factorization of form $A/\norm{A}_{F}= P \circ Q$. In this work, we provide efficient QRAM data structures for a more general class of factorizations. Let $s_{p}(A) := \max_{i \in [n]} \norm{a_{i}}_{p}^{p}$ where $\norm{a_{i}}_{p}$ denotes the $\ell_{p}$ norm of the $i$-th row of $A$, then we have the following result. 
\begin{result} \label{two} 
[Theorem \ref{sveplus1}] 
For all $p \in [0,1]$ there are efficient QRAM data structures for storing $A\in \R^{m\times n}$, such that a quantum algorithm with access to such data structures can perform SVE for $A$ with error $\delta$ in time $\widetilde{O}(\mu(A)/\delta)$ for $\mu(A) := \left(\sqrt{s_{2p}(A) s_{2(1-p)}(A^{T} )}\right)$. 
\end{result} 

We next sketch how an SVE algorithm can be used to obtain a linear system solver. Let us first assume that the matrix $A$ is positive semidefinite,  solving a linear system in $A$ reduces to performing SVE, then applying a conditional rotation by an angle proportional to the inverse of each singular value, and erasing the estimate given by the SVE algorithm. 
\al{ 
&\sum_i \beta_i \ket{v_i}\ket{\overline{\sigma}_i} \ket{0} \rightarrow \nl 
&\sum_i \beta_i \ket{v_i}\ket{0} \en{ \frac{\sigma_{min}}{\overline{\sigma}_i}\ket{0} +  \sqrt{1 - \frac{\sigma^2_{min}}{\overline{\sigma}^2_i} } \ket{1} }
} 
Postselecting on the last register being $\ket{0}$, one gets a good approximation to the desired output $\ket{A^{-1}b}= \sum_{i} \frac{\beta_{i}}{\sigma_{i}} \ket{v_{i}}$. 
To complete the analysis of the linear system solver, we boost the success probability for the SVE (which is $O(\frac{1}{\kappa^{2}(A)})$) using amplitude amplification by repeating the procedure $O(\kappa(A))$ times. Further, we show in order to obtain a $\delta$-approximate 
output for the linear system solver the SVE precision should be $O(\delta/\kappa(A))$. A similar analysis applies to matrix multiplication. The same steps can be used for general matrices using the procedure for recovering the sign of the eigenvalues given in \cite{WZP17}. If the matrix $A$ is singular we can achieve the guarantee $\norm{\ket{z} - \ket{A^{+}b}} \leq \delta$ for the Moore-Penrose pseudo-inverse, that 
is we invert only the non zero eigenvalues. 

There are different QRAM data structures for storing $A$ that can achieve $\mu(A)=\norm{A}_{F}$ by \cite{KP16} and $\mu(A) = \sqrt{s_{2p}(A) s_{2(1-p)}(A^{T} )}$ for every value of $p\in [0,1]$ by Result \ref{two}. If we are allowed two passes over the stream of the matrix entries $A$, in the first pass we can compute the minimum of the quantities $\norm{A}_{F}$ and $\sqrt{s_{2p}(A) s_{2(1-p)}(A^{T} )}$ (for an $O(1)$ sized set $\mathcal{P}$ of values for $p$ in $[0,1]$) and determine the optimal data structure which is constructed in the second pass. As the two-pass processing requires linear time, this algorithm is in the QRAM data structure model defined in Definition \ref{d3}. 
\begin{result} \label{r3} 
[Theorem \ref{lqmat}] There exists an algorithm in the QRAM data structure model that given $\ket{b}$, outputs a quantum state $\ket{z}$ with $\norm{\ket{z} - \ket{\mathcal{A}b}} \leq \delta$ for $\mathcal{A} \in \{A, A^{-1}\}$ with running time $\tilde{O}(\frac{\kappa^2(A) \mu(A)}{\delta})$ for $\mu(A) := \min_{p\in \mathcal{P}} \left( \norm{A}_{F}, \sqrt{s_{2p}(A) s_{2(1-p)}(A^{T} )}\right)$ where $\mathcal{P}$ is a set of values in $[0,1]$ with $|\mathcal{P}| = O(1)$. 
\end{result}

If we are restricted to a single pass over the matrix entries, we can construct two data structures corresponding to $\mu(A)=\norm{A}_{F}$ and $\mu(A)=s_{1}(A)$ for $p=1/2$ with a constant overhead, at the end of the pass we will know which data structure achieves the smaller value for $\mu$. These two data structures are expected to be the most useful in practice as they cover the case of low rank matrices and those with bounded $\ell_{1}$ norm. 

Let us now compare our linear system solvers in the QRAM data structure model to the HHL algorithm and its improvements. The HHL algorithm assumes a weaker input model where the transformation $O_{A}: \ket{i,j, 0} \to \ket{i,j, a_{ij}}$ can be carried out efficiently. We note that $O_{A}$ can be implemented efficiently for matrix $A$ stored in the QRAM but also without the QRAM if the matrix $A$ is well structured. However, the matrices arising in machine learning applications are not well structured necessitating the use of the stronger QRAM data structure model for these applications. 

The HHL algorithm has running time $\tilde{O}(s(A)^{2}\kappa(A)^{2}/\delta)$ if the absolute values of the eigenvalues of $A$ lie in the interval $[1/\kappa, 1]$ \cite{H14} and where $s(A)$ is the sparsity (the maximum number of non zero entries in a row), $\kappa(A)$ the condition number for the matrix $A$ and $\delta$ is the approximation error. Subsequent works have improved the running time of quantum linear system solvers to 
$\tilde{O}(s(A)\kappa(A) \log(1/\delta) )$ \cite{A12, CKS15}. In the case of dense matrices, these algorithms require time linear in the dimension of the matrix. 

Quantum linear system solvers in the QRAM data structure model have sub-linear running time even for dense matrices. Instead of sparsity, their running time depends on the parameter $\mu(A)$ defined in Result \ref{r3}. 

Let us first consider dense matrices $A$ such that the absolute values of the eigenvalues lie in the interval $[1/\kappa, 1]$. The sparsity $s(A)=\Omega(n)$ while $\mu(A) \leq \norm{A}_{F} \leq \sqrt{n}$. The linear system solver in QRAM data structure 
model obtains a worst case quadratic speedup over the HHL algorithm using the sparse access model as observed in \cite{WZP17}. As $\norm{A}_{F} = ( \sum_{i} \sigma_{i}^{2} )^{1/2} \leq \sqrt{rk(A)}$, it also achieves an exponential speedup for dense matrices with rank is poly-logarithmic in the matrix dimensions. Let us now see an example where Result \ref{r3} provides an exponential speedup 
over previous algorithms. The example is $A=I+J/n$ where $J$ is the all ones matrix, in this case we have $s(A)=\Omega(n)$, $\norm{A}_F=\Omega(\sqrt{n})$ and $\mu(A) \leq s_{1}(A)=O(1)$, demonstrating the speedup provided for Result \ref{r3}.

Let us next consider the case of symmetric sparse matrices $A$ such that all entries have absolute value bounded by $1$. In this case, we have $s_{1}(A) \leq s(A)$ for all $A$. With this scaling, we can compare with the quantum linear system solver \cite{CKS15} in the HHL input model with running time depends on $s(A)$, the QRAM data structure based approach for this case has running time depending on $s_{1}(A) \leq s(A)$. Lastly, we note that $\mu(A)=\Omega(\sqrt{n})$ for some matrices, meaning that the QRAM data structure based linear system solvers do not provide exponential savings for all matrices.

We also note that a quantum walk that can be used for linear systems with scaling $\mu(A) = \norm{ |A|}_{2}$ is known \cite{M90}, where $|A|$ is obtained by taking absolute values of the entries of $A$. However, an efficient QRAM data structure to implement this walk is not known, such an algorithm would involve updating the singular vectors corresponding to the largest singular value for $|A|$ with cost $\text{polylog}(n)$. 
Finally we note that the running time dependence for the QRAM data-structure based linear system solver has been further improved to linear in $\kappa(A)$ and then the error to $\log (1/\delta)$ in the recent work \cite{CGJ18}.

\subsection{Quantum iterative methods}\label{one3} 
Before explaining the quantum iterative method, we provide some necessary information about classical iterative methods. 

{\it Classical iterative methods for empirical risk minimization.} \label{erm} 
We consider classical iterative methods in the framework of empirical risk minimization \cite{BCN16}, where we are given $m$ examples from a training set $(x_{i}, y_{i})$ with variables $x_{i} \in \R^{n}$ and outcome 
$y_{i} \in \R$. The model is parametrized by $\theta \in \R^{n}$ and is obtained by minimizing the following objective function, 
\als{ 
F(\theta) = \frac{1}{m} \sum_{i \in [m]} \ell( \theta, x_{i}, y_{i}) + R(\theta).  
} 
The loss function $\ell( \theta, x_{i}, y_{i})$ assigns a penalty when the model does not predict the outcome $y_{i}$ well for the example 
$(x_{i}, y_{i})$ while the regularization term $R(\theta)$ penalizes models with high complexity.

The first order method for problems in this framework is called gradient descent. The algorithm starts with  $\theta_0 \in \R^{n}$, and for $\tau$ steps updates $\theta$ via the following update rule: 
\als{ 
\theta_{t+1} = \theta_{t} + \alpha \nabla F(\theta_t) 
} 
In the end, for a large class of loss functions $\ell( \theta, x_{i}, y_{i})$ it outputs $\theta_\tau$ which is guaranteed to be close to the solution for sufficiently large $\tau$. The running time is $\tau C_S$, where $C_S$ is the cost of a single step, in other words it is the cost of the update.

An important subclass of the empirical loss minimization framework is when the gradient is an affine function, as for
weighted least squares and ridge regression problems. More generally, the gradient is an affine function for quadratic optimization problems of the form $\min_{x \in \R^{n}} x^{T} Ax + b^{T}  x +c$ for some $A \in \R^{n\times n}, b\in \R^{n}, c \in \R$. In these cases, the method starts with some $\theta_0$ and for $t\geq 0$ updates it via an update rule of the form, 
\[
\theta_{t+1} = \theta_t + \alpha r_{t} 
\]
where $\alpha$ is a scalar that denotes the step size and $r_t= A\theta_t+b := L(\theta_t)$. It is easy to see that this also implies that $r_{t+1}=S(r_{t}) $ for a linear operator $S$. 
Indeed, 
\al{ 
r_{t+1}&= L(\theta_{t+1}) = L(\theta_t+\alpha r_t)= A(\theta_t+\alpha r_t) + b \nl 
& = r_t+\alpha Ar_t := S(r_t) 
} 

\noindent The final state of a linear update iterative method can hence be written as
\all { 
\theta_\tau &= \theta_0 + \alpha \sum_{t=0}^{\tau - 1} r_t \nl 
&= \theta_0 + \alpha L(\theta_0) + \alpha \sum_{t=1}^{\tau - 1} S^t(r_0).
} {three} 
where $S^t$ is the operator that applies $S$ for $t$ time steps and $S^0$ is the identity operator.  

We are slightly going to change notation in order to make the presentation of the quantum algorithm clearer. We rename $\theta_0$ as $r_0$, which means that  $L(\theta_0)$ is renamed as $L(r_0)$. 
This way, we have
\al{ 
\theta_\tau &= r_0+\alpha L(r_0) + \alpha \sum_{t=1}^{\tau-1} S^t(L(r_0)) \nl 
&=  r_0 + \alpha \sum_{t=1}^{\tau} S^{t-1}(L(r_0)).
} 
Without loss of generality we assume the initial point has unit norm, i.e. $\norm{r_0}=1$. 

{\it A quantum gradient descent algorithm with affine updates.}
Let us make things simpler for this exposition by looking at the case where we take $r_0=0$ and $\alpha=1$, meaning that we want to output the unnormalized state $\ket{\theta_\tau}=\sum_t r_t$. We only make this assumption here for conveying the main ideas and not in the proofs where we address the most general case. 

Imagine that there was a procedure that performs the following mapping perfectly
\[ \ket{t}  \norm{\theta_t} \ket{\theta_t}  \rightarrow \ket{t+1} \norm{\theta_{t+1}} \ket{\theta_{t+1}}
\]
Then, our task would be easy, since applying this unitary $\tau$ times would provide us with the desired state $\ket{\theta_\tau}$. However, this is not the case. The mapping $\theta_t$ to $\theta_{t+1}$ is not even a unitary transformation, and the norm of $\theta_{t+1}$ can be larger than the one of $\theta_t$. Even so, imagine one could in fact perform this mapping with some ``probability" (meaning mapping $\theta_t$ to some state $(\beta \norm{\theta_{t+1}} \ket{\theta_{t+1}}\ket{0} + \sqrt{1-\beta^2}\ket{G}\ket{1})$, for some garbage state $G$). The main problem is that one cannot amplify this amplitude, since the state $\ket{\theta_{t+1}}$ is unknown, being the intermediate step of the iterative method, and in the quantum case we only have a single copy of this state. Hence, the issue with the iterative method is that one needs to perform $\tau$ sequential steps, where each one may have some constant probability of success without the possibility of amplifying this probability. The probability of getting the desired final state is proportional to the product of the success probabilities for each step, which drops exponentially with the number of steps $\tau$. This is also the reason previous attempts for a quantum gradient descent algorithm break down after a logarithmic number of steps \cite{RSPS16}. 

Here we manage to overcome this obstacle in the following way. 
The first idea is to deal with the vectors $r_t$ instead of the $\theta_t$'s, since in this case, we know that the norm of $r_{t+1}$ is smaller than the norm of $r_t$. Our goal would be to find a unitary mapping that, in some sense, maps $r_t$ to $r_{t+1}$. Again, there is the problem that the norms are not equal, but in this case, since the norm of $r_{t+1}$ is smaller, we can make it into a unitary mapping by adding some garbage state. Indeed, we define the quantum step of the quantum iterative method via the following unitary 
\[
\ket{t} \norm{r_t}\ket{r_t}  \stackrel{{V}}{\rightarrow}\ket{t+1} ( \norm{r_{t+1}} \ket{r_{t+1}}\ket{0} + \ket{G}\ket{1}),
\]
where the norm of the garbage state is such that the norm of the right hand side is equal to $\norm{r_t}$. Note that the above vectors are not unit norm but $V$ is still length preserving.
Since we are dealing with linear updates, the above transformation is a matrix multiplication, and we use the SVE procedure to perform it with high accuracy.

The second idea is noticing that our goal now is not to obtain the final state $r_\tau$, but the sum of all the vectors $\sum_t r_t$. Let us see how to construct this efficiently.
Given a procedure for performing one step of the iterative method as above, we design another procedure $U$ that given as input a time $t$ and the initial state $r_0$ can map $r_0$ to $r_t$. We do this by applying $t$-times the unitary $V$, conditioned on the first register. In other words, we can perform the mapping 
\[
\ket{t,r_0}  \stackrel{U}{\rightarrow}\ket{t} ( \norm{r_t} \ket{r_{t}}\ket{0} +\ket{G}\ket{1}).
\] 
Note that if the cost of $V$ is $C_V$, then the cost of $U$ can be at most $\tau C_V$ by applying $V$ sequentially, where the error also increases by a factor of $\tau$. We will see that in fact for some cases we can implement $C_U$ in time $O(C_V+\log \tau)$.

We are now ready for the last step of the algorithm that consists in starting with a superposition of time steps from 0 to $\tau$ and applying $U$, in order to get a superposition of the form 
\[
\frac{1}{\sqrt{\tau}} \sum_t \ket{t} \ket{r_0}  \rightarrow \frac{1}{\sqrt{\tau}}\sum_t \ket{t}( \norm{r_t} \ket{r_{t}}\ket{0} +\ket{G}\ket{1}).
\] 
Then, we can ``erase'' the time register by performing a Hadamard on the first register and accepting the result when the first register is 0.  In other words, we are having a state of the form
\[
\frac{1}{{\tau}}\sum_t \norm{r_t} \ket{r_{t}}\ket{0} + \ket{G'}\ket{1}
\]
Using Amplitude Amplification, we can get the desired state $\frac{1}{\norm{\theta_\tau}}\sum_t \norm{r_t} \ket{r_{t}}$, in overall time $O(\frac{\tau}{\norm{\theta_\tau}})$ times the cost of applying the unitary $U$, and since in our applications $\norm{\theta_{\tau}}=\Omega(1)$ we get the efficient quantum gradient descent algorithm.
\begin{result}
[Informal version of Theorem \ref{aim}, Proposition \ref{48}] Given a unitary $V$ that approximately applies one step of the iterative method with affine update rules in time $C_V$, there exists a quantum algorithm that performs $\tau$ steps of the iterative method in equation \eqref{three} and outputs a state close to $\theta_\tau$, in time at most $O(\tau C_V)$. 
\end{result}

\noindent Theorem \ref{aim} provides a precise statement, the cost is $O(\tau C_{U})$ for a unitary $U$ with cost at most $O(\tau C_{V})$. Proposition \ref{48} shows that $C_{U}=O( C_V)$ for the iterative method in equation \eqref{three} where the matrix $A$ is fixed. If the matrices used for the iterative method are variable then the running time is $O(\tau^{2} C_V)$, this case occurs for the quantum stochastic gradient descent algorithm. 
For the classical iterative method in equation \eqref{three}, our running time is linear in the number of steps times the cost of taking one step of the iterative method. A single step of the iterative method in the quantum case uses 
the SVE algorithm and its cost can be considerably smaller than that for classical algorithms.

Our algorithm does not create all the intermediate states $\theta_t$, which we do not know how to achieve with non-negligible probability for large $\tau$. Instead, we observe that the final state $\theta_\tau$ is equal to the sum of all the update states $r_t$ and then we try to create the sum of these states. We first go to a superposition of all time steps from $0$ to $\tau$ and then conditioned on the time being $t$ we apply coherently $t$ updates to the initial state $r_0$ in order to create a sort of ``history'' quantum state. This is reminiscent of the ``history" states in Kitaev's QMA-completeness proof for the Local-Hamiltonian problem \cite{KSV02}. Last, erasing the register that keeps the time can be done in time linear in the number of time steps, which is still efficient. History states have been used in the past in many different scenarios, for example, a construction using history states to compose solutions of the HHL algorithm was developed independently for quantum differential equation solvers \cite{BCOW17}. We believe this technique will have further applications in quantum machine learning. 

\subsection{ Applications of quantum gradient descent}\label{one4} 

\noindent {\it Positive semidefinite linear systems.}
The simplest application for quantum gradient descent is an iterative algorithm for solving positive semidefinite linear systems. The iterative algorithm method for this application incurs an extra factor of $\kappa(A)$ compared to a direct method using Result \ref{r3}. It is still interesting to have an alternative method for linear systems even if the asymptotic running time is slightly worse. 

Classically, in machine learning settings with noisy data one uses gradient descent methods for solving linear systems instead of the direct method as it offersa tradeoff between time and accuracy \cite{BCN16}. If the data is noisy, we do not require high accuracy solutions to the linear system, and gradient descent provides the solutions
faster than the direct method. 

For this application, we are given a positive semidefinite matrix $A$ and a vector $b$, output a state close to 
$\ket{A^{-1}b}$. Let $A|b$ denote the matrix with row $b$ added to $A$. We have the following result, 
 \begin{result}  \label{four} 
 [Theorem \ref{qlinsys}] 
Let $A$ be a positive semidefinite matrix and $b$ a vector, there is an iterative QRAM-based quantum algorithm that outputs a state $\ket{z}$ such that $\norm{\ket{z} - \ket{A^{-1}b}} \leq 2\delta$ with expected running time $\tilde{O}(\frac{\kappa(A)^3 \mu(A|b)}{\delta})$.
\end{result}
\noindent We will see in the next application to least squares problems that the quantum iterative method offers additional advantages.

\noindent
{\it Stochastic gradient descent for weighted least squares.}
Our main application of the quantum iterative method is to the weighted least squares problem. Here, we are given a matrix $X$ of examples and a corresponding vector $y$ of labels, as well as a vector $w$ of weights, and the goal is to find $\theta$ that minimizes the 
squared loss $\sum_{i} w_{i} (y_{i} - x_{i}^{T}  \theta)^{2}$. 
The problem has a closed form solution given by 
\[
\theta = (X^{T}  W X)^{-1} X^{T}  W y, 
\]  
so it can also be solved using a direct method. Quantum algorithms for unweighted least squares problems with a polynomial dependence on sparsity in the HHL input model have been described \cite{WBL12}. Other approaches to least squares via column sampling have also been considered \cite{LZ17}. We extend these works in two ways: first, using our improved SVE algorithm we can perform matrix multiplication and inversion efficiently for a larger class of matrices; second, we solve the weighted version of the problem. 

More importantly, we are able to give an iterative 
stochastic gradient method for this problem which has many advantages in practical settings (see for example \cite{BCN16} for a more detailed discussion). The least squares problem is used in practice for regression or data fitting, where in many cases the data matrix is skewed in shape since the number of data points is an order of magnitude larger than the dimension of the data points. In such cases, it is prohibitive to perform classical linear algebra operations using the entire data set, and moreover, due to redundancy in the data, the gradient can be estimated efficiently over small batches. 

For these reasons, the gradient is estimated over randomly sampled batches of the training set, an approach which is called stochastic gradient descent. This way, the stochastic gradient descent avoids having to perform linear algebra operations on large matrices, which would be the case if we were to solve the problem directly or using gradient descent. 
Our quantum iterative method can also be used to perform stochastic gradient descent for the above problems.

There are two significant advantages of the quantum stochastic gradient descent algorithm. First, the stochastic gradient descent algorithm has much reduced quantum memory requirements compared to the HHL algorithm or directly solving the linear system in the QRAM data structure model. Similar to the classical setting, the data is split randomly into batches, and for every step of the iterative method, only one batch needs to be in quantum memory. 
This would be very useful in a situation where the size of QRAMs is limited by the difficulty of maintaining coherent superpositions over a large number of qubits. In such a situation, it may be possible to build a 1Mb QRAM but impossible to combine 10 such devices to obtain a 10Mb QRAM. Due to the difficulties in designing QRAMs we believe that this is a significant advantage offered by the quantum stochastic gradient descent algorithm. 

In order to reduce the slowdown caused by the loading of data into the QRAM, the quantum stochastic gradient descent method can be used with multiple quantum memory devices each with limited capacity. In this setting, the algorithm  would perform a computation using one of the memory devices, while data is loaded concurrently into the others. Such memory architectures can simulate a larger memory using smaller sized devices, 
they are also widely used in classical computing. Again, the actual quantum speedup achieved for this setting would depend on a number of hardware parameters, the advantage of the quantum stochastic gradient descent method is to provide greater flexibility and to enable quantum data processing applications for large sized datasets. Further analyses of such architectures can be interesting future work.

The main technical difference between our algorithms for linear systems and least squares is that in the latter, one needs to perform matrix operations for a matrix that is not a priori stored in memory. More precisely, we have in memory the matrix $X$, the diagonal matrix $W$ and a vector $y$ and we need to perform matrix multiplication with the matrix $(X^{T}  W X)^{-1}$ and also create the vector $X^TWy$. We do this by first using the known weight vector $W$ to update the QRAM data structure then invoking the matrix multiplication procedure with the updated data structure.
\begin{result} \label{five} 
[Theorem \ref{qwlsq}] 
Let $X$ be an arbitrary matrix, $W$ a diagonal matrix and $y$ a vector, there is an iterative QRAM-based quantum algorithm that outputs $\ket{z}$ such that $\norm{\ket{z} - \ket{(X^{T}  W X)^{-1} X^TWy}} \leq 2\delta$ in expected time $\tilde{O}(\frac{\kappa(X^TWX)^3. 
\mu(\sqrt{W} X|y)}{\delta})$.
\end{result} 

\subsection{Related Work and Discussion} 
In this section, we discuss how this work relates to more recent works on quantum linear algebra and provide some perspective on directions for future work. 

The block encoding framework has recently emerged \cite{CGJ18} as a general framework for embedding arbitrary matrices into unitary matrices. 
An efficient block encoding for a matrix $A$ is an efficiently implementable unitary matrix $U$ such that $U= \left ( \begin{matrix} A/\alpha &, . \\ . &, . \end{matrix} \right )$
where $\alpha \geq \norm{A}_{2} $ is a scaling factor. Given efficient block encodings for input matrices, there are quantum linear system solver with running time $\widetilde{O}(\mu \kappa \log (1/\epsilon)$ \cite{CGJ18}. Efficient block encodings are also required for quantum singular value transformation \cite{GLSW18} where a function of the singular values is to be applied to the amplitudes. 

The results in this paper can be viewed as providing a family of algorithms (parametrized by $\mu$) for creating efficient block encodings for arbitrary matrices in the QRAM data structure model. They can be used to construct the block encodings required for the state of the art quantum linear system solvers \cite{CGJ18, GLSW18} and therefore occur 
as subroutines in a variety of applications of quantum linear algebra.

A large class of quantum machine learning algorithms involving low-rank input matrices have been dequantized following Tang's breakthrough work on quantum-inspired recommendation systems \cite{T18}. The dequantization methods are based on low-rank decomposition are not applicable to the quantum linear systems for sparse matrices or those with bounded $\ell_{1}$ norm. Quantum machine learning algorithms remain interesting in spite of the dequantization results due to the large polynomial overheads  
of the dequantized results  over the quantum algorithm. For example, for recommendation systems with $m$ users, $n$ products, and $k$ "types of users", the quantum algorithm has running time $\widetilde{O}(\frac{k^{1/2}}{\epsilon} \text{polylog}(mn))$ compared to the dequantized running time $\widetilde{O}(\frac{k^{12}}{\epsilon^{12}} \text{polylog}(mn))$. 
In fact the dequantized algorithm is also much slower than the classical algorithm that has running time $\widetilde{O}(mk)$ on instance sizes of practical interest.
It remains an open question to obtain classical algorithms with smaller polynomial overheads in $k$ and $1/\epsilon$. 

Iterative methods are ubiquitous in classical optimization, however developing quantum analogs of these iterative methods is significantly more challenging. The iterative method that we provide here are able to achieve a polynomial scaling in the number of iterative steps for the restricted class of functions where the gradient is an affine function. Finding other frameworks for quantum iterative methods remains a direction for future work. More recently, we presented a quantum interior point method for linear and semidefinite programs \cite{KP18} which is an example of a second order quantum iterative method.

Finding further applications of the quantum stochastic gradient descent algorithm is another direction for future work. This algorithm can be particularly useful in settings where the linear system matrix is not available a priori but only revealed partially as the algorithm makes progress. Such settings are common for dynamic programming type algorithms in reinforcement learning \cite{SB98}. They arise for instance for the problem of playing video games where the states of the game and the actions are revealed to the algorithm only when the algorithm moves ahead into the game and are not known in the very beginning of the game. Classical stochastic gradient descent can be used for optimizing the strategy in such a game. 
This offers an example where a direct application of the HHL algorithm would not suffice as the entire matrix is not known in advance, but the quantum stochastic gradient descent algorithm would enable such an application since at every step the gradient would be estimated via the information that has already been revealed. 

\section{Quantum Preliminaries}

\subsection{Quantum Algorithms}
We will use phase estimation and variants of amplitude amplification that we recall below. The time required to implement a unitary operator $U$ will be denoted by $T(U)$. 
\begin{theorem} \label{pest} 
[Consistent phase estimation,  \cite{T13}] Let $U$ be a unitary operator with eigenvectors $\ket{v_j}$ and eigenvalues $e^{ \iota \theta_{j}}$ for $\theta_{j} \in [-\pi, \pi]$. There exists a quantum algorithm with running time $O( T(U) \frac{\log n }{\epsilon} )$ that transforms $\ket{\phi} = \sum_{j \in [n]} \alpha_{j} \ket{v_{j}} \to \sum_{j \in [n]} \alpha_{j} \ket{v_{j}}\ket{ \overline{\theta_{j}} }$ where $\overline{\theta_{j}}$ is an estimate such that $|\overline{\theta_{j}} - \theta_{j} | \leq \epsilon$ for all $j\in [n]$ with probability at least $1-1/\emph{poly}(n)$. 
\end{theorem} 
\noindent The difference with phase estimation result in \cite{K95} is that the usual phase estimation does not output a fixed estimate $\overline{\theta_{j}}$, the estimate can in some cases assume one 
of two possible values, while here we require that the estimate be fixed with high probability. The estimate however is not deterministic, it depends on the random shifts used in the consistent phase estimation 
algorithm. 
An alternate consistent eigenvalue estimation procedure is 
sketched in \cite{A12, KLLP18}. We also need a version of amplitude amplification and estimation \cite{BHMT00}, 
\begin{theorem} \label{tampa} 
[Amplitude amplification and estimation \cite{BHMT00}] If there is unitary operator $U$ such that $U\ket{0}^{l}= \ket{\phi} = \sin(\theta) \ket{x, 0} + \cos(\theta) \ket{G, 0^\bot}$ where 
$\ket{G}$ is an arbitrary garbage state and $\ket{0^\bot}$ is any state orthogonal to $\ket{0}$ then  $\sin^{2}(\theta)$ can be estimated within error $(1 \pm \epsilon)$ in time $O(\frac{T(U)}{\epsilon \sin(\theta)})$ and 
$\ket{x}$ can be generated in expected time $O(\frac{T(U)}{\sin (\theta)})$. 
\end{theorem} 
\noindent Note that for this statement, if the amplitude amplification succeeds then we measure $\ket{0}$ and obtain the exact state $\ket{x}$, there is no approximation error involved. 
Last we provide a simple claim that shows that if two unnormalized vectors $\phi, \tilde{\phi}$ are close to each other, then their normalized versions 
$\ket{\phi}, \ket{\tilde{\phi}}$ are also relatively close.

\begin{claim}\label{unnorm}
Let $\theta$ be the angle between vectors $\phi, \tilde{\phi}$ and assume that $\theta< \pi/2$. Then, $\norm{ \phi -  \tilde{\phi} } \leq \epsilon$ implies $\norm{ \ket{\phi} - \ket{ \tilde{\phi} }} \leq \frac{ \sqrt{2} \epsilon }{ \norm{\phi} }.$
\end{claim}
\begin{proof}
 
We bound the $\ell_{2}$ distance $\norm{ \ket{\phi} - \ket{ \tilde{\phi}}}$ using the following argument. Let $\theta$ be the angle between $\phi, \tilde{\phi}$. For the unnormalized vectors we have $\norm{ \phi -\tilde{\phi} } \leq \epsilon$, and assuming that $\theta< \pi/2$ we have $\epsilon  \geq \norm{\phi} \sin (\theta)$.  The distance between the normalized states can thus be bounded as, 
\all{ 
\norm{ \ket{\phi} - \ket{ \tilde{\phi} }}^{2} = (2 \sin (\theta/2))^{2} \leq 2\sin^{2} (\theta) \leq \frac{2  \epsilon^{2}} {  \norm{\phi}^{2}} 
} {norm0} 
\end{proof}

\section{The Quantum Gradient Descent algorithm}
In this section, we define a quantum step of the quantum gradient descent in the case of a linear update rule and then describe the full quantum procedure that performs the quantum iterative method. 

\subsection{The quantum step}

We assume that the classical iterative method has an update rule of the form, 

\[
\theta_\tau =  r_0 + \alpha \sum_{t=1}^{\tau} S^{t-1}(L(r_0)) = r_{0} + \alpha \sum_{t \in [\tau]} r_{t} 
\]
 for an affine operator $L$, a linear, contracting operator $S$, and an arbitrary initial vector $r_0$ with $\norm{r_0}=1$. This is the case, for example, for solving linear systems or least squares. 

First, we define the notion of an approximate quantum step of the quantum iterative method. Let us denote by $\tau$ the number of steps of the classical iterative algorithm, and let $\tau +1= 2^\ell$ (if not just increase one can $\tau$ to the next power of 2). A contracting linear operator $L$ is an operator such that $\norm{Lx} \leq \norm{x}$ for all $x$.

\begin{defn}\label{31} 
Let $L$ be an affine operator, $S$ be a linear, contracting operator, $r_0$ be an initial vector and $r_{t} =S^{t-1}(L(r_0))$ be the residual vectors. 
An $(\epsilon, \delta)$-approximate quantum step corresponding to a classical iterative method with update rule $\theta_\tau =  r_0 + \alpha \sum_{t=1}^{\tau} S^{t-1}(L(r_0))$ 
 is a transformation $V$ such that for any $1 \leq t \leq \tau-1$,
\begin{align*}
V & : \ket{0} \ket{r_0} \ket{0}\rightarrow \ket{1} \left( \alpha \norm{\tilde{L}(r_0)} \ket{\tilde{L}(r_0)} \ket{0} + \ket{G_1}\ket{1} \right) \\
&: \ket{t} \norm{r_t} \ket{r_t} \ket{0} \rightarrow  \\
&\ket{t+1} \left( \norm{\tilde{S}(r_{t})}  \ket{\tilde{S}(r_{t})}\ket{0} + \ket{G_{t+1}}\ket{1} \right),
\end{align*}
where $\ket{G_{t+1}}$, for $0\leq t \leq \tau -1$ is an unnormalised garbage state, 
$\tilde{L}$ is an approximation to $L:r_0 \rightarrow r_{1}$, 
$\tilde{S}$ is an approximation to $S:r_t \rightarrow r_{t+1}$, 
such that $\norm{ L(r_0) - \tilde{L}(r_0)} \leq \epsilon $ with probability $\geq 1-\delta$,  and for each $t\in [\tau-1]$ we have $ \norm { S(r_{t}) -  \tilde{S}(r_{t}) }  \leq \epsilon $ with probability $\geq 1-\delta$. 
\end{defn}

Notice that $\norm{\tilde{L}(r_0)}$ might be larger than 1, but by taking $\alpha$ a small constant we have $\alpha \norm{\tilde{L}(r_0)}  \leq 1$. The way we defined $V$, it is norm preserving, but the vectors that we wrote do not have unit norm. We can, of course, normalize them by dividing both sides by $\norm{r_t}$. We prefer this notation in order to resemble more the classical mapping of the unnormalized vectors $r_t \rightarrow r_{t+1}$.

We also note that the transformation $V$ will be eventually implemented using a singular value estimation procedure what succeeds with probability $1-1/\text{poly}(n)$, thus $\delta=1/\text{poly}(n)$ is the probability with which the estimates 
for $L(r_0)$ and $S(r_{t})$ are correct. We can define the following procedure $U$ similar to the ideal case.

\begin{claim}\label{U}
Given access to an $(\epsilon, \delta)$-approximate quantum step procedure $V$ with cost $C_V$, there exists an $(\epsilon, \delta)$-approximate quantum multistep procedure $U$ with cost at most $O(\tau C_V)$, such that for any $t \in [\tau]$,
\begin{align*}
& U  : \ket{0} \ket{0} \ket{0} \ket{r_0} \ket{0} \rightarrow \ket{0} \ket{0} \ket{0} \ket{r_0} \ket{0} \\
&\;\;\;: \ket{t} \ket{0} \ket{0} \ket{r_0} \ket{0} \rightarrow \\  
&\ket{t} \ket{0} \left( \alpha \norm{\tilde{S}^{t-1}(\tilde{L}(r_{0}))}  \ket{t} \ket{\tilde{S}^{t-1}(\tilde{L}(r_{0}))}\ket{0} + \ket{G'_t}\ket{1} \right),
\end{align*}
where $\ket{G'_t}$ is an unnormalised garbage state, and with probability at least $(1-t\delta)$ it holds that  $ \norm{  S^{t-1}(L(r_0)) - \tilde{S}^{t-1}(\tilde{L}(r_{0})) } \leq t \epsilon$, 
where $S, L$ are as in Definition \ref{31}. 
\end{claim}
\noindent Note that out of the five registers used in this claim, registers 1 and 3 store the time step, 4 stores the quantum state for the iterative method while 2 and 5 are control qubits or flags. Moreover, we take into account the fact that the error in $U$ increases by a factor of $t$.

\begin{proof} 
We define the operator $W$ on four registers, such that if the control register is 0, then it applies a $V$ on the other three registers and then a CNOT to copy the last register into the control register. If the control register is 1, then it does nothing. Namely
\begin{align*}
W & : \ket{0} \ket{0} \ket{r_0} \ket{0} \rightarrow \\
&\ket{0} \ket{1} \alpha \norm{\tilde{L}(r_0)} \ket{\tilde{L}(r_0)} \ket{0} + \ket{1}\ket{1}\ket{G_1}\ket{1}, \\
& :  \ket{0} \ket{t} \norm{r_t} \ket{r_t} \ket{0} \rightarrow \\
&\ket{0} \ket{t+1}  \norm{\tilde{S}(r_{t})} \ket{\tilde{S}(r_{t})}\ket{0} +  \ket{1} \ket{t+1} \ket{G_{t+1}}\ket{1}\\
&  \;\;\;\forall  t \in [1,\tau-1], \\
& :  \ket{1} \ket{t} \ket{b} \ket{1} \; \rightarrow \ket{1} \ket{t} \ket{b} \ket{1} \\
& \;\;\;  \forall t \in [0,\tau-1]. 
\end{align*}

We define the following procedure $U$ that acts as identity for $t=0$ and for $t \in [\tau-1]$ it does the following:
\all{ 
& \ket{t} \ket{0} \ket{0}\ket{r_0} \ket{0}  \stackrel{W^t}{\rightarrow}  \ket{t} W^t (\ket{0} \ket{0} \ket{r_0}  \ket{0}) = \nl 
& \ket{t}   \ket{0} \ket{t} \alpha \norm{\tilde{S}^{t-1}(\tilde{L}(r_{0}))}  \ket{\tilde{S}^{t-1}(\tilde{L}(r_{0}))}\ket{0} + \nl 
& \ket{t} \ket{1}  \sum_{i=1}^{T}  \ket{i} \ket{G_i} \ket{1}   \stackrel{CNOT_{5,2}}{\rightarrow} \nl
& \ket{t} \ket{0} \left( \alpha  \norm{\tilde{S}^{t-1}(\tilde{L}(r_{0}))}  \ket{t} \ket{\tilde{S}^{t-1}(\tilde{L}(r_{0}))} \ket{0} +  \ket{G'_t}\ket{1} \right)\notag 
} {eqn11} 
The equality in the above equation follows from the definition of $W$. The cost of applying $W$ is cost $C_{V}$, thus the entire procedure can be performed in time $O(\tau C_{V})$. We prove the properties by induction on $t$. For $t=1$ we use the definition of the quantum step $V$ and the property holds. Assume it holds for $t-1$, i.e. with probability $1-(t-1)\delta$ we have 
$\norm{ S^{t-2}(L(r_{0})) -  \tilde{S}^{t-2}(\tilde{L}(r_0)) } \leq (t-1)\epsilon $. 

Then, we have
\als{ 
&\norm{ S^{t-1}(L(r_0)) - \tilde{S}^{t-1}(\tilde{L}(r_{0})) }  \leq \nl 
& \norm{ S(S^{t-2}(L(r_0))) - S( \tilde{S}^{t-2}(\tilde{L}(r_{0}))) } \nl 
 &  +  \norm{ S( \tilde{S}^{t-2}(\tilde{L}(r_{0}))) - \tilde{S}( \tilde{S}^{t-2}(\tilde{L}(r_{0}))) }\nl 
 & \leq   \norm{ S^{t-2}(L(r_0)) - \tilde{S}^{t-2}(\tilde{L}(r_{0})) } \nl
 &  +   \norm{ S( \tilde{S}^{t-2}(\tilde{L}(r_{0}))) - \tilde{S}( \tilde{S}^{t-2}(\tilde{L}(r_{0}))) }
} 
where we used the fact that ${S}$ is contractive.
Also, by definition of the iterative step, 
with probability $(1-\delta)$ we have $ \norm{ S( \tilde{S}^{t-2}(\tilde{L}(r_{0}))) - \tilde{S}( \tilde{S}^{t-2}(\tilde{L}(r_{0}))) } \leq \epsilon $ and with probability $1-(t-1)\delta$, by induction hypothesis, we have $\norm{ S^{t-2}(L(r_{0})) -  \tilde{S}^{t-2}(\tilde{L}(r_0)) } \leq (t-1)\epsilon $. Hence overall, with probability at least $1-t\delta$, we have 
\[
 \norm{ S^{t-1}(L(r_0)) - \tilde{S}^{t-1}(\tilde{L}(r_{0})) }  \leq   t \epsilon .
\]
\end{proof}

\subsection{The Quantum Iterative Method: general case} \label{io} 

The main part of the quantum iterative method is the efficient construction of a unitary $Q$ defined below using the results from Claim \ref{U}. We then use amplitude amplification and estimation on $Q$ to improve the running time of the method. \\

\noindent
{\bf The Quantum Iterative Method}
Use Amplitude Amplification and Estimation with unitary $Q$ 
\[
Q : \ket{0}^\ell \rightarrow \frac{1}{T} \ket{\tilde{\theta}_\tau}\ket{0} + \ket{G}\ket{0^\bot}
\]
where $\ket{G}$ is a garbage state and $\ket{0^{\perp}}$ denotes a state that is orthogonal to $\ket{0}$, to output 
$\ket{\tilde{\theta}_\tau}$ and an estimate for $\norm{\tilde{\theta}_\tau}$. The parameter $T$ is taken to be $\frac{ \tau+1}{\norm{\tilde{\theta}_\tau}} $. 
We next provide a procedure for implementing $Q$.   \\

\noindent
{\bf Implementation of the unitary }$Q : \ket{0}^\ell \rightarrow \frac{1}{T} \ket{\tilde{\theta}_\tau}\ket{0} + \ket{G}\ket{0^\bot}$\\
\begin{enumerate}
\item Create the state $\frac{1}{\sqrt{\tau+1}} \sum_{t=0}^{\tau} \ket{t} \ket{0}\ket{0}  \ket{r_0} \ket{0}$

\item Apply the unitary procedure $U$ from Claim \ref{U} and trace out the second register to get
\als{ 
& \frac{1}{\sqrt{\tau+1}}  \sum_{t=1}^{\tau}  \ket{t} \big ( \alpha \norm{\tilde{S}^{t-1}(\tilde{L}(r_{0}))}  \ket{t} \ket{\tilde{S}^{t-1}(\tilde{L}(r_{0}))}\ket{0} + \nl 
& \ket{G'_t}\ket{1} \big ) +  \frac{1}{\sqrt{\tau+1}} \ket{0}\ket{0}  \ket{r_0} \ket{0}. 
} 
\item Conditioned on the last register being $0$, perform a $CNOT_{1,2}$ to erase the second copy of $t$ and then by exchanging the place of the second and third register we get 
\als{ 
& \frac{1}{\sqrt{\tau+1}} \sum_{t=1}^{\tau}  \ket{t} \big ( \alpha \norm{\tilde{S}^{t-1}(\tilde{L}(r_{0}))}   \ket{\tilde{S}^{t-1}(\tilde{L}(r_{0}))}\ket{0}\ket{0} + \nl 
& \ket{G'_t}\ket{1} \big) + \frac{1}{\sqrt{\tau+1}} \ket{0} \ket{r_0} \ket{0} \ket{0}.  \nl 
} 

\item Conditioned on the last register being $0$ perform a Hadamard on all qubits of the first register and then by exchanging the place of the first and second register we get
\als{ 
& \frac{1}{\tau+1}  \sum_{y=0}^{\tau} \Big(   \ket{r_0} + \alpha \sum_{t=1}^{\tau}  (-1)^{y \cdot t} \norm{\tilde{S}^{t-1}(\tilde{L}(r_{0}))}  \nl 
& \ket{\tilde{S}^{t-1}(\tilde{L}(r_{0}))}  \Big ) \ket{y}\ket{0}  \ket{0} +  \ket{G^{\prime\prime}}\ket{1}  \nl 
& =\frac{\norm{\tilde{\theta}_\tau}}{\tau+1} \Big(  \frac{1}{\norm{\tilde{\theta}_\tau}}  \Big( \ket{r_0} + \alpha \sum_{t=1}^{\tau}  \norm{\tilde{S}^{t-1}(\tilde{L}(r_{0}))}  \nl 
& \ket{\tilde{S}^{t-1}(\tilde{L}(r_{0}))}  \Big ) \Big)  \ket{0}\ket{0} \ket{0}  + \ket{G} ( \ket{0}\ket{0}\ket{0})^\bot  \nl 
&= \frac{1}{T}\ket{\tilde{\theta}_\tau}\ket{0} + \ket{G}\ket{0^\bot}. 
}  
with 
$\ket{\tilde{\theta}_\tau} =  \frac{1}{\norm{\tilde{\theta}_\tau}}  \Big( \ket{r_0} + \alpha \sum_{t=1}^{\tau}  \norm{\tilde{S}^{t-1}(\tilde{L}(r_{0}))}  
 \ket{\tilde{S}^{t-1}(\tilde{L}(r_{0}))} \Big)  $ 
and $T = \frac{ \tau+1}{\norm{\tilde{\theta}_\tau}}$.
\end{enumerate}

\subsection{Analysis}\label{analysis}
The main result for this section that proves the correctness and bounds the running time for the quantum iterative method is the following: 
\begin{theorem} \label{aim} 
Given unitary $U$ with cost $C_{U}$ that implements an $(\epsilon, \delta)$-approximate quantum multistep procedure (as in Claim \ref{U}) for an iterative method with affine update rules and $\tau$ steps, there is a quantum algorithm with running time $O(T(C_{U}+ \log \tau))$ with $T=\frac{ \tau+1}{\norm{\tilde{\theta}_\tau}} $ that outputs a state $\ket{\tilde{\theta}_\tau}$ such that, 
$$\norm{\ket{\theta_\tau} - \ket{\tilde{\theta}_\tau}} \leq \frac{\sqrt{2} \alpha \tau^2 \epsilon}{\norm{\theta_\tau}}. $$
\end{theorem} 
\noindent We prove separately the statements about the correctness and the running time for the quantum iterative method below, together with some additional comments about the applications. 
\subsubsection{Correctness}

We calculate the $\ell_{2}$ distance between the final state $\ket{\tilde{\theta}_\tau}$ and the correct state which is given by
$\ket{{\theta}_\tau} =  \frac{1}{\norm{{\theta}_\tau}}  \left( \ket{r_0} + \alpha \sum_{t=1}^{\tau}  \norm{{S}^{t-1}({L}(r_{0}))}   \ket{{S}^{t-1}({L}(r_{0}))} \right)$.
We first compute the non-normalised distance when a $(\epsilon, \delta)$ approximate step is used for implementing $Q$, 

\als{ 
 \norm{\theta_\tau - \tilde{\theta}_\tau} & \leq \alpha \sum_{t=1}^{\tau} \norm{S^{t-1}(L(r_0)) - \tilde{S}^{t-1}(\tilde{L}(r_0))}  \nl 
& \leq \alpha  \sum_{t=1}^{\tau} t \epsilon \leq \alpha \tau^2 \epsilon
} 
Then by Claim \ref{unnorm} we can also bound the normalized distance, 
\[
\norm{\ket{\theta_\tau} - \ket{\tilde{\theta}_\tau}} \leq \frac{\sqrt{2} \alpha \tau^2 \epsilon}{\norm{\theta_\tau}}. 
\]
\noindent This proves the statement about correctness for Theorem \ref{aim}. 

Further, we note that Amplitude Estimation outputs an estimate $\norm{\tilde{\theta}_\tau}$ such that $\frac{ \norm{\tilde{\theta}_\tau}} { \norm{\theta_\tau} } \in (1\pm \xi)$ with a constant overhead (where $\xi$ is a small constant). In the applications we will consider, $\norm{\theta_\tau}$ is at least $\Omega(1)$ and at most $O(\tau)$ and $\alpha=O(1)$. Hence, again, by taking $\epsilon=O(\frac{1}{\tau^2})$ appropriately small, we can ensure that the approximation error for the $\ell_{2}$ distance is less than 
a suitably small constant.

\subsubsection{Running time} \label{itertime}

The expected running time is the expected running time of Amplitude Amplification (which is the same as that of Amplitude Estimation), which is $T$ times the cost of implementing the unitary $Q$, which is $O(C_U+\log\tau)$. Overall, the expected running time is $O(T (C_U+\log\tau))$ as claimed in Theorem \ref{aim}. 

Let us make some additional remarks. In our applications we will have $\norm{\theta_\tau}=\Omega(1)$, which also implies that
$\norm{\tilde{\theta}_\tau}\geq \norm{\theta_\tau} - \alpha \tau^2 \epsilon=\Omega(1)$ for appropriately small $\epsilon$.
Hence, the running time for the applications will be $O(\tau (C_U+\log\tau))$. In many applications including psd linear systems and weighted least squares problems, 
each step of the iterative method involves multiplication by a fixed matrix. In this case, we can implement $U$ with cost $C_U=O(C_V+\log \tau)$ as shown in 
Proposition \ref{48}) and hence get an overall running time $O(\tau (C_V+\log \tau))$. However, the cost $C_{U}$ can be $O(\tau C_{V})$ for the case when each step of the iterative method involves multiplication by a different matrix. In this case the worst-case running time is $O(\tau^{2} C_V)$, this running time occurs in our application to stochastic gradient descent for the weighted least squares problem. 

For the quantum iterative method to approximate the output state to error $\delta'$, we would need the error $\epsilon$ in implementing $V$ to be $O( \frac{\delta'}{\tau^{2}})$, which would make the error in $U$ at most $O( \frac{\delta'}{\tau})$. We implement $V$ using our linear system solver (Theorem \ref{qlinsys}) and therefore achieve a running time of $O(\tau (\tau^{2} \mu(A)) /\delta)$, which is superlinear in $\tau$.

The running time of linear system solvers in the QRAM data structure model has been improved to $O(\mu(A) \kappa(A) \text{polylog}(1/\epsilon))$ in the recent work \cite{CGJ18}. More precisely, they presented a quantum linear system solver with running time  $O(\mu(A) \kappa(A) \text{polylog}(1/\epsilon))$  given a unitary block encoding for $A$, and also showed that our results on QRAM data structures presented here are equivalent to the construction 
of unitary block encodings with parameter $\mu(A)$. If we use the improved linear system solver, the running time would be $O(\tau \mu(A) \log (\tau^{2}/\delta) )$, that is the dependence on the number of steps is indeed linear.

\section{Improved quantum algorithms for matrix multiplication and linear systems} 
In Section \ref{ds} and \ref{isve} we generalize the data structure for state preparation and the quantum algorithm used for singular value estimation that we had proposed 
in \cite{KP16}. The generalized singular value estimation algorithm has a faster running time for several classes of matrices. We use the improved singular value estimation algorithm for solving quantum linear systems and quantum matrix multiplication in Section \ref{lsmul}. Finally, in Section \ref{istep} we show how to implement a single step of the quantum iterative method. 

\subsection{A generalized state preparation data structure} \label{ds} 
The QRAM data structure in \cite{KP16} enabled the efficient  preparation of vector states $\ket{v}$ for vectors $v$ stored in the QRAM. We provide a more general data structure 
that enables the efficient preparation of more general normalized states corresponding to the rows/columns of a matrix. We begin by defining the notion of a normalized state. 

\begin{defn} The normalized vector state corresponding to vector $x \in \R^{n}$ and $M \in \R$ such that $\norm{x}^{2}  \leq M$ is the quantum state $\ket{\overline{x}} = \frac{1}{\sqrt{M}} \en{ \sum_{i \in [n]} x_{i} \ket{i} + 
(M-\norm{x}^{2} )^{1/2} \ket{n+1} } $.  
\end{defn} 
\noindent We recall that in the QRAM data structure model, the entries of the matrix $A$ arrive in an online manner and are stored in a  
data structure. The insertion and update times for the data structure are poly-logarithmic per entry, so the time required to construct the data structure is $O(w \log^{2}(mn) )$ where $w$ is the number of non zero entries in $A$. We have the following theorem that describes the efficient QRAM data structure as in Definition \ref{d2}, this data structure will be used to obtain the  improved singular value estimation algorithm. 
\begin{theorem} \label{dsplus} 
Let $A \in \R^{m\times n}$ and $M= \max_{i \in [m]} \norm{a_{i}}^{2}$. There is an efficient QRAM data structure for storing matrix entries $(i, j, a_{ij})$ such that access to this data structure allows a quantum algorithm to implement the following unitary 
in time $\widetilde{O}(\log (mn) )$.  
\all{ 
U \ket{  i, 0^{\lceil \log (n+1) \rceil} } & = \ket{i} \frac{1}{\sqrt{M}} \big (  \sum_{j \in [n]} a_{ij} \ket{j} + \nl 
&(M - \norm{a_{i}}^{2})^{1/2}\ket{n+1}  \big ) 
}  {nscreate} 
\end{theorem} 
\begin{proof} 
The data structure maintains an array of $m$ binary trees $B_{i}, i \in [m]$ one for each row of the matrix. The leaf node $j$ 
of tree $B_{i}$, if present, stores $(a_{ij}^{2}, sign(a_{ij}))$. An internal node $u$ stores the sum of the values of the leaf nodes
in the subtree rooted at $u$. In addition, there is an extra node $M$ that, at any instant of time stores the maximum row norm $M=\max_{i \in [m]} \norm{a_{i}}^{2} $ for the 
matrix $A$ currently stored in the data structure. 

The data structure is initially empty, and the value stored in node $M$ is $0$. We next describe the update when entry $(i, j,a_{ij})$ is inserted or updated, note that the 
insertion and update must be carried out in time $O(\text{polylog}(n))$ for the data structure to be efficient. After giving the update procedure, we provide the 
efficient implementation for the unitary $U$ in \eqref{nscreate} given access to the data structure. 

The update algorithm on receiving input $(i, j, a_{ij})$ creates leaf node $j$ in tree $B_{i}$ if not present and updates it otherwise. Deletion is a special case of update and is 
equivalent to the case $(0,i,j)$. The algorithm updates the value of all nodes in the path between the leaf and the root of the tree.
The update requires time $O(\log^{2} (mn))$ as at most $O(\log n)$ nodes on the path from node $j$ to the root in the tree $B_{i}$ 
are updated and each update requires time $O(\log (mn))$ to find the address and the value of the node being updated. At the end of each update 
the root of $B_{i}$ stores the squared norm $\norm{a_{i}}^{2}$ for the vector stored in $B_{i}$. The algorithm then compares $M$ with $\norm{a_{i}}^{2}$ and 
updates the maximum if $\norm{a_{i}}^{2} > M$, this additional step requires time $O(\log n)$. It follows that the data structure is efficient and that after the 
update, the value stored in the node $M$ is $\max_{i \in [m]} \norm{a_{i}}^{2}$. 

In order to implement $U$ we first perform a controlled rotation on the second register using the values $M$ and $\norm{a_{i}}$ which is stored at the root of $B_{i}$. We also introduce a tag qubit and 
tag the part of the state on the second register with value $\ket{n+1}$, 
\all{ 
\ket{  i, 0^{\lceil \log (n+1) \rceil } }  \to  \ket{i} \frac{1}{ \sqrt{M}} \big( \norm{a_{i}} \ket{0^{\lceil \log (n+1) \rceil }} \ket{0} + \nl 
(M - \norm{a_{i}}^{2} )^{1/2}  \ket{n+1} \ket{1} \big) 
 } {tag} 
We then proceed similarly to the construction in  \cite{KP16}. Let $B_{i,k}$ be the value of an internal node $k$ of tree $B_{i}$ at depth $t$. We apply a series of conditional rotations to the second register, conditioned on the first register being $\ket{i}$ and the first $t$-qubits of the second register being $\ket{k}$ and the tag qubit being $\ket{0}$, the rotation applied is: 
\[ 
\ket{i} \ket{k} \ket{0} \to \ket{i}\ket{k} \frac{1}{\sqrt{ B_{i, k}}} \left( \sqrt{ B_{i,2k} } \ket{0} + \sqrt{ B_{i,2k+1} } \ket{1} \right) 
\] 
We take positive square roots except for the leaf nodes where the sign of the square root is the same as $sign(a_{ij})$ of the entry stored at the leaf node. 
The tag qubit is uncomputed after all the conditional rotations have been performed by mapping $\ket{n+1}\ket{1}$ to $\ket{n+1}\ket{0}$. 

Correctness follows since conditioned on the tag qubit being $\ket{0}$ the conditional rotations produce the state $\frac{1}{ \sqrt{M}\norm{a_{i}}} \sum_{j} a_{ij} \ket{j}$
and the amplitude for the tagged part is $\sqrt{ (M - \norm{a_{i}}^{2})/M } $ matching the amplitudes in equation \eqref{nscreate}. The time for implementing $U$ is 
$\widetilde{O}(\log (mn))$ as the number of quantum queries to the data structure is $\log ^{2}(mn)$ and each query takes poly-logarithmic time. 
\end{proof}

\subsection{Improved Singular Value Estimation} \label{isve} 
We first recall the definition of singular value estimation from the introduction, 
\begin{definition} \label{sve} 
Let $A= \sum_{i \in [k]} \sigma_{i} u_{i} v_{i}^{T} $ be the singular value decomposition for $A \in \R^{m \times n}$ for $k=\min(m,n)$ and let $\delta>0$. The singular value estimation (SVE) problem with error $\delta$ is given $\ket{b}= \sum_{i\in [k]} \beta_i \ket{v_i}$, to perform the mapping  
\[
 \sum_{i\in [k]}  \beta_i \ket{v_i} \ket{0} \rightarrow \sum_{i \in [k]}  \beta_i \ket{v_i}\ket{\overline{\sigma}_i},
\]
such that $|\overline{\sigma}_i - \sigma_i| \leq \delta$ for all $i \in [k]$ with probability $1-1/poly(n)$. 
\end{definition}

\noindent The theorem below provides a generalized singular value estimation algorithm that extends the algorithm 
from \cite{KP16} and the quantum walk algorithms used for linear systems \cite{CKS15}. 
\begin{theorem} \label{sveplus} 
Let $A \in \R^{m\times n}$ be a matrix and suppose there exist $P, Q \in \R^{m\times n}$ and $\mu >0$ such that
$\norm{p_{i}}_{2} \leq 1 \; \forall i \in [m], \;\norm{q^{j}}_{2} \leq 1 \; \forall j \in [n]$ and
\all{ 
A/\mu = P \circ Q.  
} {hadamard} 
If unitaries $U: \ket{i} \ket{0^{\lceil \log (n+1) \rceil}} \to \ket{i} \ket{\overline{p}_{i}}$ and $V:  \ket{0^{\lceil \log (m+1) \rceil}}  \ket{j} \to \ket{\overline{q}^{j}} \ket{j}$ can be implemented in time $O(\log^{2} (mn))$ then there is a quantum algorithm for SVE with error $\delta$ in time $\widetilde{O}(\mu/\delta)$. 
\end{theorem} 
\begin{proof} 
Let $\overline{P}, \overline{Q} \in \R^{(m+1) \times (n+1)}$ be matrices with rows and columns respectively equal to the normalized states  
$\overline{p}_{i},  \overline{q}^{j}$ for $i\in [m],  j\in [n]$ and an additional row or column $\overline{p}_{m+1} = e_{m+1}, \overline{q}^{n+1} = e_{n+1}$.  
Let $\overline{A}= \left (
\begin{matrix} 
A &0 \\ 
0 &\mu  
\end{matrix} \right )$ be an extension of $A$ of size $(m+1) \times (n+1)$ so that the factorization $\overline{A}/\mu = \overline{P} \circ \overline{Q}$ holds. 

As $\overline{A}$ is a block diagonal matrix, its singular value decomposition is given by $\sum_{i} \sigma_{i} \overline{u}_{i} \overline{v}_{i}^{T}  + \mu e_{m+1}^{T}  e_{n+1}$ where $\sigma_{i}$ are singular values for $A$ and $\overline{u}_{i}, \overline{v}_{i}$ are obtained by appending an 
additional $0$ coordinate to the singular vectors $u_{i}, v_{i}$ of $A$. Define the operators $\widetilde{P} \in \R^{(m+1)(n+1) \times (m+1)}, 
\widetilde{Q} \in \R^{(m+1)(n+1) \times (n+1)}$ as, 
$\widetilde{P} \ket{i} = \ket{i} \ket{ \overline{p}_{i}}$ and  $ \widetilde{Q} \ket{j}  = \ket{ \overline{q}^{j} } \ket {j}$.

The columns of $\widetilde{P}, \widetilde{Q}$ are orthogonal unit vectors so we have $\widetilde{P}^{T}  \widetilde{P} = I_{m+1}$ and $\widetilde{Q}^{T}  \widetilde{Q}=I_{n+1}$. Multiplication by $\widetilde{P}, \widetilde{Q}$ can be implemented efficiently using the unitaries $U$, $V$ in the theorem statement. 
The unitary $\widetilde{U}$ on input $\ket{i} \ket{0^{\lceil \log (n+1) \rceil}}$ acts as $U$ for $i\in [m]$ and maps $\ket{m+1}\ket{0^{\lceil \log (n+1) \rceil}} \to \ket{e_{n+1}}$, it can easily be implemented using $U$. We illustrate below for multiplication by $\widetilde{P}$, 
$ \ket{z} \to \ket{ z, 0^{\lceil \log (n+1) \rceil} } \xrightarrow{\widetilde{U}} \sum_{i \in [m+1]} z_{i} \ket{ i, \overline{p}_{i} } = \ket{ \widetilde{P} z}$. 
Multiplication by $\widetilde{Q}$ can be implemented similarly using $\widetilde{V}$, thus the reflections  $2\widetilde{P}\widetilde{P}^{T}  -I$ and $2\widetilde{Q}\widetilde{Q}^{T}  -I$ can be performed in time $\widetilde{O}(\log (mn))$. 

Finally, the factorization $\widetilde{P}^{T} \widetilde{Q} = \overline{A}/\mu$ implies that the unitary $W= (2\widetilde{P}\widetilde{P}^{T}  -I). 
(2\widetilde{Q}\widetilde{Q}^{T}  -I)$ has eigenspaces $Span(\widetilde{P} \overline{u}_{i}, \widetilde{Q} \overline{v}_{i})$ with eigenvalues $e^{ \iota \theta_{i}}$ such that 
$\cos(\theta_{i}/2) = \sigma_{i}/\mu$. This relation between the eigenvalues of $W$ and the singular values of $A$ is known (for example, see \cite{KP16}, Lemma 5.3). 
Phase estimation for $W$ on $\ket{\widetilde{Q}\overline{v}_{i}}$ recovers an estimate $\overline{\theta_{i}}$ for $\theta_{i}$ up to additive error $\delta/\mu$ in time $\widetilde{O}(\mu/\delta)$, which provides an estimate $\cos(\overline{\theta_{i}}/2)$ for $\sigma_{i}$ within additive error $O(\delta)$. The generalized singular value estimation is stated as Algorithm \ref{gensve}.

\begin{algorithm}[H] 
\caption{Generalized quantum singular value estimation.} \label{gensve}
\begin{algorithmic}[1]
\REQUIRE $A \in \R^{m\times n}$, $x \in \R^{n}$, efficient implementation of unitaries $U, V$. 
Input state $\ket{x}$ and precision parameter $\epsilon>0$.  \\
\begin{enumerate} 
\item Let $\ket{\overline{x}} = \sum_{i \in [n]} \alpha_{i} \ket{\overline{v}_{i}}$ be the decomposition of the input in the basis of 
singular vectors for $\overline{A}$. Create state $\ket{0^{\lceil \log (m+1) \rceil }, \overline{x}}$ and apply unitary $\widetilde{V}$ to 
obtain $\ket{\widetilde{Q}\overline{x}} = \sum_{i} \alpha_{i} \ket{\widetilde{Q}\overline{v}_{i}}$. \\
\item Perform phase estimation with precision $2\epsilon >0$
on input $\ket{\widetilde{Q}\overline{x}}$ for the unitary $W$ in Theorem \ref{sveplus}. 
\item Compute $\overline{\sigma_{i}} = \cos(\overline{\theta_{i}}/2) \mu(A)$ where 
 $\overline{\theta_{i}}$ is the estimate from phase estimation, and uncompute the 
output of the phase estimation to obtain $\sum_{i} \alpha_{i} \ket{\widetilde{Q}\overline{v}_{i}, \overline{\sigma_{i}} } $.   \nl 
\item Apply the inverse of $\widetilde{V}$ to multiply the first register with the inverse of $\widetilde{Q}$ and obtain $\sum_{i} \alpha_{i} \ket{\overline{v}_{i} } \ket{\overline{\sigma_{i}}}$. \nl 
 \end{enumerate} 
\end{algorithmic}
\end{algorithm}
\noindent Note that in step 1 we work with $\overline{x} = (x,0) \in \R^{n+1}$, this is the same as working with $\ket{x}$ as 
adding an additional $0$ coordinate does not change the corresponding vector state.

\end{proof} 
\noindent The main theorem for this section is obtained by using the data structure from Theorem \ref{dsplus} to efficiently implement the unitaries  $U,V$ required for the generalized singular value estimation procedure. 

\begin{theorem} \label{sveplus1} 
For (i) $\mu(A) = \sqrt{s_{2p}(A) s_{2(1-p)}(A^{T} )}$ for all $p\in [0,1]$ and (ii) $\mu(A) = \norm{A}_{F}$, there are efficient QRAM data structures for storing $A\in \R^{m\times n}$ such that a quantum algorithm with access to such data structures can perform SVE for $A$ with error $\delta$ in time $\widetilde{O}(\mu(A)/\delta)$. 
\end{theorem} 
\begin{proof} 

We first prove part (i). Theorem \ref{sveplus} holds for any choice of $P, Q$ such that $A/\mu = P \circ Q$ provided the unitaries $U$ and $V$ can be implemented efficiently in time $\widetilde{O}(\log (mn))$, that is if the normalized states corresponding to the rows of $P$ and the columns of $Q$ can be prepared efficiently. For a given $p\in [0,1]$, the data structure in Theorem \ref{dsplus} allows us to prepare the normalized states for the following choice of $P$ and $Q$, 
\all{ 
p_{ij} = \frac{ sgn(a_{ij}) a_{ij}^{p} } { \max \norm{a_{i}}_{2p}^{2p} },  \;\;\;\;\; q_{ij} = \frac{ a_{ij}^{1-p} } { \max \norm{a^{j}}_{2(1-p)}^{2(1-p)}}
} {fact} 
Indeed, in order to implement the unitaries $U$ and $V$ corresponding to this choice of $P, Q$, we create two copies of the data structure in Theorem \ref{dsplus} that respectively store the rows and the columns of $A$. Given entry $(i, j, a_{ij})$ instead of $a_{ij}$, we store $sgn(a_{ij}) a_{ij}^{p}$ and $a_{ij}^{1-p}$. 
The normalization factor $\mu$ for this factorization is $\mu_{p}(A) = \sqrt{ s_{2p}(A) s_{2(1-p)}(A^{T} )}$, where $s_{p}(A) = \max_{i \in [m]} \norm{a_{i}}_{p}^{p}$ and $s_{p}(A^{T} )= \max_{j \in [n]} \norm{a^{j}}_{p}^{p}$.

Part (ii) follows from the data structure used for SVE algorithm given in \cite{KP16}, it corresponds to the 
factorization $A/\norm{A}_{F}= P \circ Q$ with $p_{ij} = \frac{ a_{ij} } { \norm{a_{i}} }$ and $q_{ij} =  \frac{ \norm{a_{i}}  } { \norm{A}_{F} }$. 
\end{proof} 
Theorem \ref{sveplus1} gives different efficient data structures that provide an SVE algorithm with running time $\widetilde{O}(\mu(A)/\delta)$ for different values of $\mu$. We show next that in the QRAM data structure model, one can achieve the minimum $\mu(A)$ over a subset of these data structures. 

\begin{theorem} \label{optqds} 
Let $\mathcal{P}$ be a set of values in $[0,1]$ such that $|\mathcal{P}|=O(1)$. There is algorithm in the QRAM data structure model (Definition \ref{d3}), that performs SVE for $A \in \R^{m\times n}$ in time $\widetilde{O}(\mu(A)/\delta)$ 
with $\mu(A) = \min_{p \in \mathcal{P}} \left( \norm{A}_{F}, \sqrt{s_{2p}(A) s_{2(1-p)}(A^{T} )}\right)$. 
\end{theorem} 
\begin{proof} 
Let $m$ be the number of non zero entries in $A$, we make two passes over the entries $(i,j, a_{ij})$. In the first pass we find the optimal value 
$\mu(A) = \min_{p\in [0,1]} \left( \norm{A}_{F}, \sqrt{s_{2p}(A) s_{2(1-p)}(A^{T})} \right)$. The cost for this is $O(|\mathcal{P}|m)$ as we need to 
updating the $2|\mathcal{P}|$ values $s_{2p}(A), s_{2(1-p)}(A^{T})$ each time we process a matrix entry $(i,j, a_{ij})$. At the end of the 
first pass we know the data structure that corresponds to the optimal value of $\mu(A)$. In the second pass we stream over the matrix entries and contrucct the data structure corresponding to 
$\mu(A)$ as in Theorem \ref{sveplus1}. 

The overall complexity is linear in $m$ as $|\mathcal{P}|$ is a constant, the pre-processing step therefore satisfies the requirements for the QRAM data structure model in Definition \ref{d3}. 
Given access to the optimal data structure, Theorem \ref{sveplus1} implies that there is a quantum algorithm for SVE with error $\delta$ in time $\widetilde{O}(\mu(A)/\delta)$ for $\mu(A) = \min_{p \in \mathcal{P}} \left( \norm{A}_{F}, \sqrt{s_{2p}(A) s_{2(1-p)}(A^{T} )}\right)$. 
\end{proof}

 If we are restricted to a single pass over the matrix entries, we can construct two data structures for $A$ for $\mu(A)=\norm{A}_{F}$ and $\mu(A) = s_{1}(A)$ (corresponding to $p=1/2$) with a constant overhead and then use the data structure that achieves the smaller value for $\mu(A)$. This construction covers the case of sparse and low rank matrices that arise often in practice and will be the one that we shall use for our iterative method. 
 
 \begin{theorem} \label{simps} 
 There is algorithm in the QRAM data structure model (Definition \ref{d3}), that makes a single pass over the entries of $A \in \R^{m\times n}$ in the pre-processing step 
 and performs SVE for $A$ in time $\widetilde{O}(\mu(A)/\delta)$ for $\mu(A) = \min \left( \norm{A}_{F}, s_{1}(A) \right)$. 
 \end{theorem} 
 \noindent For our applications, one could also use the more general Theorem \ref{optqds} but we use Theorem \ref{simps} instead as it simplifies the analysis, involves only one pass over the matrix entries in the pre-processing  step and covers the useful cases of bounded $\ell_{1}$ norm and low-rank $A$ for which there is a significant quantum speed-up.

\subsection{Quantum matrix multiplication and linear systems} \label{lsmul} 
We provide algorithms for quantum linear systems and quantum matrix multiplication using the improved singular value estimation algorithm. 
We will see that once we perform singular value estimation for a matrix $A$, then multiplication with the matrix consists of a conditional rotation by an angle proportional to each singular value. Similarly, solving the linear system corresponding to the matrix $A$ is multiplication with the inverse of $A$, in other words, a conditional rotation by an angle proportional to the inverse of each singular value of $A$. 

The two algorithms are therefore very similar. We will also extend our matrix multiplication algorithm, i.e., the application of a linear operator, to the case of an affine operator, namely, given matrix $A$ and vector $b$ in memory, the algorithm maps any state $\ket{x}$ to a state close to $\ket{Ax+b}$. Last, we discuss briefly the cases for which our algorithm improves upon the running time of existing quantum linear system solvers.

If $A \in \R^{m\times n}$ is a rectangular matrix, then multiplication by $A$ reduces to multiplication by the square symmetric matrix 
$A'= \left (
\begin{matrix} 
0 &A \\ 
A^{T}  &0  
\end{matrix} \right )$ as $A' (0^{m}, x) = (Ax, 0)$. Therefore, without loss of generality we restrict our attention to symmetric matrices for the quantum matrix multiplication problem.  We state the quantum matrix multiplication Algorithm \ref{qmat} for a positive semidefinite matrix as in applications of the iterative method. We note that linear systems for general symmetric matrices are not much harder than the case described in Algorithm \ref{qmat}.

More precisely, the SVE procedure estimates the absolute value of the eigenvalues $|\lambda_{i}|$ for a symmetric $A$, and hence we also need to recover the sign of the $\lambda_{i}$ to perform matrix multiplication or to solve linear systems. The sign can be recovered by performing singular value estimation for the matrices $A, A + \mu I$ and comparing the estimates $|\overline{\lambda}_{i}|, |\overline{\lambda_{i} + \mu}|$. The second estimate is larger if $\lambda_{i}>0$, otherwise the first estimate is larger. We refer to \cite{WZP17} for details. We do not use this more general linear system solver in this paper, linear system solvers for positive semidefinite matrices described next suffice for 
our applications. 

\begin{algorithm}
\caption{Quantum matrix multiplication/linear systems.} \label{qmat}
\begin{algorithmic}[1]
\REQUIRE Matrix $A \in \R^{n\times n}$ stored in the data structure in Theorem \ref{optqds}, such that eigenvalues of 
$A$ lie in $[1/\kappa, 1]$. Input state $\ket{x}=\sum_{i} \beta_{i} \ket{v_{i}}\in \R^{n}$, where $v_{i}$ are right singular vectors for $A$. 
\STATE Perform singular value estimation with precision $\epsilon_{1}$ for $A$ on $\ket{x}$ to obtain $\sum_{i} \beta_{i} \ket{v_{i}} \ket{\overline{\lambda}_{i}}$. 
\STATE Perform a conditional rotation and uncompute the $SVE$ register to obtain the state: \nl 
(i)  $\sum_{i} \beta_{i} \ket{v_{i}} ( \overline{\lambda_{i}}  \ket{0} + \gamma \ket{1})$ for matrix multiplication.  \nl 
(ii) $\sum_{i} \beta_{i} \ket{v_{i}} ( \frac{1} {\kappa \overline{\lambda_{i}} }  \ket{0} + \gamma \ket{1})$ for linear systems. \nl 
\STATE Perform Amplitude Amplification with the unitary $V$ implementing steps $1$ and $2$, to obtain (i) $\ket{z} = \sum_{i} \beta_{i} \overline{\lambda_{i}}  \ket{v_{i}}$ or (ii) $\ket{z} = \sum_{i} \beta_{i} \frac{1}{\overline{\lambda_{i}}} \ket{v_{i}} $. \nl 
\end{algorithmic}
\end{algorithm}

\noindent 
We next analyze the correctness and provide the running time for the above algorithm.  
\begin{theorem} \label{lqmat}
Parts (i) and (ii) of Algorithm \ref{qmat} produce as output a
state $\ket{z}$ such that 
 $\norm{ \ket{\mathcal{A}x} - \ket{z}} \leq \delta$ in expected time
 $\widetilde{O}(\frac{\kappa^{2}(A) \mu(A)}{ \delta }  )$ for $\mathcal{A}= A$ and $\mathcal{A}= A^{-1}$ respectively. 
\end{theorem}
\begin{proof} 
We first analyze matrix multiplication. The unnormalized solution state is $Ax= \sum_{i} \beta_{i} \lambda_{i} v_{i}$, while the unnormalized output of step 1 of the algorithm which performs SVE to precision $\epsilon_{1}$ is $z= \sum_{i} ( \lambda_{i} \pm  \widetilde{\epsilon}_{i}) \beta_{i} v_{i}$ such that $|\widetilde{\epsilon}_{i}| \leq \epsilon_{1}$ for all $i$. As the $v_{i}$ are orthonormal, we have 
$\norm{ Ax - z} \leq \epsilon_{1} \norm{x}$ and by Claim  \ref{unnorm}, we have $\norm{ \ket{Ax} - \ket{z}} \leq \frac{\sqrt{2}\epsilon_{1}\norm{x}}{\norm{Ax}}\leq \sqrt{2} \epsilon_{1} \kappa(A)$.

We next analyze linear systems. The unnormalized solution state is $A^{-1}x= \sum_{i} \frac{\beta_{i}}{\lambda_{i}} v_{i}$. 
The unnormalized output is $z= \sum_{i} \frac{\beta_{i}}{\lambda_{i} \pm \tilde{\epsilon}_{i}} v_{i}$ for $|\widetilde{\epsilon}_{i}| \leq \epsilon_{1}$. We have
the bound  
\begin{align*}
\norm{ A^{-1} x - z}^{2} &\leq   \sum_{i} \beta_{i}^{2} \en{ \frac{1}{\lambda_{i}} - \frac{1}{ \lambda_{i} \pm \tilde{\epsilon}_{i}} }^{2} \nl 
&\leq \epsilon_{1}^{2} \sum_{i} \frac{\beta_{i}^{2}}{\lambda_{i}^{2} (\lambda_{i} - \epsilon_{1})^{2} } \leq \frac{\epsilon_{1}^{2} \kappa^{2}(A) \norm{A^{-1} x}^{2}}{ (1-\kappa(A) \epsilon_{1})^{2}} \nl 
& \leq  4\epsilon_{1}^{2} \kappa^{2}(A) \norm{A^{-1} x}^{2}
\end{align*}
assuming 
that $\kappa(A) \epsilon_{1} \leq 1/2$.  Applying Claim  \ref{unnorm} we obtain $\norm{ \ket{A^{-1} x} -\ket{ z}} \leq  2\sqrt{2}\kappa(A) \epsilon_{1}$ for $\kappa(A) \epsilon_{1} \leq 1/2$. 

We can therefore use the SVE algorithm with precision $\epsilon_{1} = \frac{\delta} { \kappa(A)} $ for both cases to obtain a solution state $\norm{ \ket{\mathcal{A}x} - \ket{z}} \leq \delta$. The success probability for step 2 of the algorithm is $\frac{1}{\kappa^{2}(A)}$ for both matrix multiplication and matrix inversion. We apply Amplitude Amplification as in Theorem \ref{tampa} with the unitary $V$ that represents the first two steps of the algorithm to obtain an expected running time $\widetilde{O}(\frac{\kappa^2 (A) \mu(A)}{ \delta } )$. 
\end{proof}

Let us now see how the linear system solver in the QRAM data structure model compares to the HHL algorithm and its improvements. The HHL algorithm requires an input model where the transformation $O_{A}: \ket{i,j, 0} \to \ket{i,j, a_{ij}}$ can be carried out efficiently. The oracle $O_{A}$ can be implemented efficiently for matrix $A$ stored in the QRAM but also without the QRAM if the matrix $A$ is well structured. The QRAM data structure model thus makes a stronger assumption than the HHL input model. It is however illuminating to see the improvements one can obtain over HHL in the stronger QRAM data structure model. 

 A natural normalization for the quantum linear system problem is to assume that $\norm{A} = 1$, so that the eigenvalues being estimated have been scaled down to quantities in $[1/\kappa,1]$. We provide an algorithm to normalize the matrix $A$ such that $\norm{A} \leq 1$ in Section \ref{sne}. 

The HHL algorithm with this scaling produces an $\delta$ approximation to the state $\ket{A^{-1} b}$ and an estimate for $\norm { A^{-1} b}$ in time $\widetilde{O}( s^{2}(A) \kappa^{2}(A)/\delta)$ where $s(A)$ is 
an upper bound on the number of non-zero entries per row. Subsequent work has improved the running time to $\widetilde{O}( s(A) \kappa(A)\log(1/\delta))$ \cite{CKS15, A12}, the recent work 
\cite{CGJ18} implies similar improvements for linear system solvers in the QRAM data structure model. 

Quantum linear system solvers in the QRAM data structure model have sub-linear running time even for dense matrices. Instead of sparsity, their running time depends on the parameter $\mu(A)$. On one hand, this factor is smaller than the Frobenius norm, for which we have $\norm{A}_{F} = ( \sum_{i} \sigma_{i}^{2} )^{1/2} \leq \sqrt{rk(A)}$. Hence, Algorithm \ref{qmat} achieves a considerable speedup for matrices whose rank is poly-logarithmic in the matrix dimensions. Moreover, while for general dense matrices the sparsity is $\Omega(n)$, we have $\mu(A) \leq \norm{A}_{F} \leq \sqrt{n}$. The QRAM data structure based linear system solvers therefore achieve a worst case quadratic speedup over linear system solvers in the HHL input model \cite{WZP17}.

In addition, the factor $\mu(A)$ is smaller than the maximum $\ell_{1}$-norm $s_1(A)$ which is smaller than the maximum sparsity $s(A)$ for the normalization $\norm{A}_{max}=1$, that is the entries of $A$ have absolute value at most $1$. This normalization is used for the linear system solver \cite{CKS15} in the HHL input model, the QRAM data structure based linear system solver presented here is an improvement as $s_{1}(A) \leq s(A)$. 
An example where the linear system solver of Theorem \ref{lqmat} is an exponential improvement over previous approaches is for real-valued matrices with most entries close to zero and a few entries close to 1, for example  $A=I+ J/n$ or a small perturbation of a permutation matrix. We have $s(A)=\Omega(n)$, $\norm{A}_F=\omega(\sqrt{n})$ and $s_1(A)=O(1)$ for small enough perturbations.

Algorithm \ref{qmat} achieves a worst case quadratic speedup over linear system solvers in the HHL input model as there are matrices with $\norm{A}=1$ for which $\mu(A)=\Omega(\sqrt{n})$. An example is a random sign matrix $A$ with $\norm{A}=1$. In this case, one can easily show that with high probability $\mu(A) = \Omega(\sqrt{n})$. 

The optimal value $\mu(A) =\min_{\mu} \{ P \circ Q = A/\mu \;|\; \norm{p_{i}}_{2} \leq 1, \norm{q^{j}}_{2} \leq 1 \}$ in Theorem \ref{sveplus} is the spectral norm of $|A|$ where $|A|$ is the matrix obtained by replacing entries of $A$ by their absolute values \cite{M90}. We recall that the matrix $A \in \R^{m\times n}_{+}$ has unique positive left and right eigenvectors, these eigenvectors are called the Perron-Frobenius eigenvectors and can be computed for example by iterating  $x_{t+1} = Ax_{t}/\norm{Ax_{t}}$. The optimal walk can be implemented efficiently if the coordinates of the Perron-Frobenius eigenvectors for $|A|$ are stored in memory prior to constructing the data structure. 

However, there are no known algorithms for computing the entries of the Perron-Frobenius eigenvector in linear time in the number of matrix entries with a poly-logarithmic update time. Hence the optimal walk cannot be implemented in the quantum data structure model. We also note that the spectral norm of $|A|$ can be much larger than the spectral norm of $A$, for example for a random sign matrix with $\pm 1$ entries the spectral norm of $A$ is $\norm{A}= O(\sqrt{n})$ but the spectral norm of $|A|$ is $\norm{|A|}=n$. Thus there are two interesting questions about quantum linear system solvers. The first is to find the optimal quantum linear system solver in the QRAM data structure model. The second is to find if there can exist more general quantum walk algorithms with $\mu(A)= \norm{A}$. The latter has also been stated as an open problem \cite{C10}.

We also note that as in the analysis of the HHL algorithm \cite{HHL09}, the parameter $\kappa$ does not have to be as big as
the actual condition number of $A$. If $\kappa$ is smaller, then it means that we invert only the well-conditioned part of the matrix
and this may be useful for some settings.

\subsubsection{Spectral Norm Estimation} \label{sne} 
We assumed above that the matrices $A$ are normalized such that the absolute value of the eigenvalues 
lie in the interval $[1/\kappa, 1]$. We provide here a quantum algorithm using the QRAM data structure given in \cite{KP16} for estimating the spectral norm, which can be used to rescale matrices so that the assumption $\norm{A} \leq 1$ is indeed valid. Note that $0 \leq \frac{ \norm{A}} {\norm{A}_{F}} \leq 1$ and that $\norm{A} = \sigma_{\max}(A)$, where 
$\sigma_{\max}$ is the largest singular value. 

\begin{algorithm}[H]
\caption{Spectral norm estimation.} \label{algsne}
\begin{algorithmic}[1]
\REQUIRE $A \in \R^{m\times n}$ stored in data structure given in \cite{KP16}. Returns an estimate for 
$\eta:= \norm{A}/\norm{A}_{F}$ with additive error $\epsilon$.  \\

\begin{enumerate} 
\item Let $l=0$ and $u=1$ be upper and lower bounds for $\eta$, the estimate $\tau = (l+u)/2$ is refined using binary search in steps 2-5
over $O(\log 1/\epsilon)$ iterations.  

\item Prepare $\ket{\phi} = \frac{1}{\norm{A}_{F}} \sum_{i,j} a_{ij} \ket{i, j} = \frac{1}{\norm{A}_{F}} \sum_{i,j} \sigma_{i}  \ket{u_{i}, v_{i}}$
and perform SVE \cite{KP16} with precision $\epsilon$ to obtain $\frac{1}{\norm{A}_{F}} \sum_{i,j} \sigma_{i}  \ket{u_{i}, v_{i}, \overline{\sigma_{i}} }$.  
where $|\overline{\sigma}_{i} - \frac{\sigma_{i}}{\norm{A}_{F}}| \leq \epsilon$. 

\item Append single qubit register $\ket{R}$ and set it to $\ket{1}$ if $\overline{\sigma}_{i} \geq \tau$ and $\ket{0}$ otherwise. Uncompute the SVE output from step 2. 

\item Perform amplitude estimation on $\frac{1}{\norm{A}_{F}} \sum_{i,j} \sigma_{i}  \ket{u_{i}, v_{i}, R}$ conditioned on $R=1$ 
to estimate $\sum_{i: \overline{\sigma}_{i} \geq \tau} \sigma_{i}^{2}/\norm{A}_{F}^{2}$ to relative error $(1\pm \delta)$.

\item If estimate in step 4 is $0$ then $u \to \tau$ else $l \to \tau$. Set $\tau = (u+l)/2$. 
\end{enumerate} 
\end{algorithmic}
\end{algorithm}
\noindent The following proposition proves correctness for Algorithm \ref{sne} and bounds its running time. Algorithm \ref{sne} provides an estimate for $\norm{A}/\norm{A}_{F}$, 
this yields an estimate for $\norm{A}$ as one can easily compute $\norm{A}_{F}$ when $A$ is being efficiently loaded into the QRAM. 
\begin{proposition} 
Algorithm \ref{sne} estimates $\norm{A}$ to error $\epsilon \norm{A}_{F}$ in time $\widetilde{O}( \log (1/\epsilon)/ \epsilon \eta)$.  
\end{proposition} 
\begin{proof} 
We show that the algorithm estimates $\eta$ within additive error $\epsilon$. In order to prove this, it suffices to show that if $|\tau - \eta|\geq \epsilon$ then the algorithm determines $sign(\tau - \eta)$ correctly.  If $|\tau - \eta|\geq \epsilon$ then the amplitude $\sum_{i: \overline{\sigma}_{i} \geq \tau} \sigma_{i}^{2}/\norm{A}_{F}^{2}$ being estimated in step 4 is either $0$ or at least $\eta^{2}$. Amplitude estimation as in Theorem \ref{tampa} yields a
non-zero estimate in the interval $(1\pm \delta)\eta$ for the latter case and thus the sign is determined correctly for constant $\delta$. 

The running time for step 2 is $\widetilde{O}(1/\epsilon)$. The amplitude estimation in step 4 requires time  $\widetilde{O}(1/\epsilon\delta \eta)$
as $T(U)= \widetilde{O}(1/\epsilon)$ and the amplitude being estimated is either $0$ or at least $\eta^{2}$. Step 4 
is repeated $\log(1/\epsilon)$ times and $\delta$ is a fixed constant, so the running time is $\widetilde{O}( \log (1/\epsilon)/ \epsilon \eta )$. 
\end{proof} 
\noindent The above algorithm can also be used to find an estimate for the condition number $\kappa$, that can be used in the linear systems solver. Let $\kappa'$ be a threshold such that $\sum_{\sigma_{i} \leq \kappa'} \sigma_{i}^{2}/\norm{A}_{F}^{2} \geq \eta^{2}$, then the algorithm (with $\ket{R}$ set to $1$ if $\overline{\sigma_{i}} \leq \tau$ in step 3) provides an additive error $\epsilon$ estimate for $\kappa'$ in time $\widetilde{O}( \log (1/\epsilon)/ \epsilon \eta)$. The spectral norm estimation procedure can be viewed as an auxiliary subroutine that is used once and thus contributes an additive term to the running time for the quantum linear system solver. We do not include this additive overhead in our running time estimates.

\subsection{The iterative step} \label{istep}  
We now show that the $(\epsilon, \delta)$-approximate quantum step for the iterative method in Definition \ref{31} can be implemented using the quantum matrix multiplication algorithm presented above. The matrix $S$ for iterative methods corresponding to linear systems and least squares is of the form $S=I - \alpha A$ for a positive semidefinite matrix $A$. Further, we can assume that the matrix $A$ is stored in the data structure of Theorem \ref{simps} and that $S$ is positive semi-definite and contractive. 

\begin{proposition} \label{iterimp} 
An $(\epsilon, 1/poly(n))$-approximate quantum step (as in Definition \ref{31}) for the iterative method with $S=I - \alpha A, L(x)=b-Ax$,  $\alpha\leq 1$, $\norm{b}=1$ and for $A| b$ stored in the data structure of Theorem \ref{simps},
can be implemented in time $\widetilde{O}( \mu(A | b) / \epsilon)$, where $A | b$ is the matrix $A$ with an extra row equal to $b$.  
\end{proposition} 

\begin{proof} 
We show how to implement the unitary $V$ that implements the $(\epsilon, 1/poly(n))$-approximate quantum step (as in Definition \ref{31}) for the iterative method, that is 
\als{ 
V &: \ket{0} \ket{r_0} \ket{0}\rightarrow \ket{1} \left( \alpha \norm{\tilde{L}(r_0)} \ket{\tilde{L}(r_0)} \ket{0} + \ket{G_1}\ket{1} \right) \nl 
&:\ket{t} \norm{r_t} \ket{r_t} \ket{0} \rightarrow  \ket{t+1} \left( \norm{\tilde{S}(r_{t})} \ket{\tilde{S}(r_{t})}\ket{0} + \ket{G_{t+1}}\ket{1} \right),
} 
where $\ket{G_{1}}, \ket{G_{t+1}}$ are unnormalised garbage states, such that with probability $\geq 1-\delta$, it holds that 
$\norm{ L(r_0) - \tilde{L}(r_0)} \leq \epsilon $ and we also have that $\norm{ S(r_{t}) -  \tilde{S}(r_{t}) } \leq \epsilon$ for all $t \in [\tau-1]$ with probability $\geq 1-\delta$. We first implement the linear part 
of $V$ that corresponds to $1 \leq t \leq \tau-1$ and then the affine part corresponding to $t=0$. We denote $r_t=\sum_{i} \beta_{i} \ket{v_{i}}$ in the basis of right singular vectors for $S$.

The linear part of $V$ is implemented by performing the singular value estimation for $A$ and then using $\overline{ \lambda_{i}} =(1-\alpha \overline{\lambda_{i}} (A))$ as estimates 
for singular values of $S$ with additive error $\epsilon$, 
\all { 
&\ket{t} \norm{r_t}\ket{r_t} \ket{0} \equiv \ket{t} \sum_{i} \beta_{i} \ket{v_{i}}  \ket{0} \to \nl 
&\ket{t+1} \sum_{i} \beta_{i} \ket{v_{i}}  \en{ \overline{\lambda_{i} }  \ket{0} + \sqrt{1- {\overline{\lambda_{i}}^2}} \ket{1} }   
} {istep1} 
As the precision for singular value estimation is $\epsilon$, the algorithm runs in time $\widetilde{O}( \mu(A) /\epsilon)$.  
For bounding the difference of the norms we observe that $\norm{ S(r_{t}) -  \tilde{S}(r_{t}) } \leq \norm{\sum_{i} \beta_{i} \tilde{\epsilon_{i}} v_{i} } \leq \alpha \epsilon \leq \epsilon$ as all the errors $\tilde{\epsilon_{i}} \leq \epsilon$ if the singular value estimation succeeds and $\norm{\beta} \leq 1$. The procedure succeeds if the singular value estimation algorithm produces the correct estimates. The success probability for the singular value estimation is $1-1/poly(n)$, thus $\delta$ can be taken to be $1/poly(n)$.

The affine part of $V$ is implemented as follows. Let $A_1= \en{ \begin{matrix} -A & b \nl 0& 0\end{matrix} }$ and $x_1= (x, 1)$ so that $A_1 x_1= (b-Ax, 0)$. Then we symmetrize $A_1$ by defining $A'= \left (
\begin{matrix} 
0 &A_1 \\ 
A_1^{T}  &0  
\end{matrix} \right )$ and $x'= (0^{n+1},x_1) $, and have $A' x' = (A_1 x_1, 0) = (b-Ax,0)$. 
The columns of $A'$ are stored in the memory data structure so we can perform SVE for $A'$. We take $x=r_0$ and use the SVE algorithm for symmetric matrices on $r'_0$ (where we add an extra factor $\alpha$ in the conditional rotation) to map it to $A' r_{0}'= (b - Ar_{0},0)$ as the last coordinates become $0$. Denote $r'_0 =  \sum_{i} \beta_{i} \ket{v'_{i}}$, where $v'_{i}$ are the eigenvectors of $A'$.
\als{ 
&\ket{0} \norm{r'_0}\ket{r_0'} \ket{0} \equiv \ket{0} \sum_{i} \beta_{i} \ket{v'_{i}}  \ket{0} \to \nl 
& \ket{1} \sum_{i} \beta_{i} \ket{v_{i}}  \en{ \alpha \overline{\lambda_{i} }(A')  \ket{0} + \sqrt{1- {\alpha^{2} \overline{\lambda_{i}}(A')^2}} \ket{1} }   
} 
If the precision for singular value estimation is $\epsilon$ then the algorithm runs in time $\widetilde{O}( \mu(A') /\epsilon)$ and the correctness analysis is the same as above. We next provide an upper bound for $\mu(A')$. We have $\norm{A'}_F \leq 2\norm{A}_F+2\norm{b}$, while $s_1(A') \leq \max(s_1(A)+\norm{b}_{\infty} ,s_1(b))$. Let's assume for simplicity that $\norm{b}=1$, which is the case in our applications, 
then the upper bound is $O(\mu(A|b))$ where $A|b$ is a matrix obtained by adding an extra row $b$ to $A$.

\end{proof}

Finally, let us see how to implement the $(\epsilon, \delta)$-approximate quantum multistep unitary $U$ in Claim \ref{U}  with cost $O(C_V+\log \tau)$ which is asymptotically the same as the cost of $V$. 
\begin{proposition} \label{48} 
Given an implementation of the unitary $C_{V}$  in Proposition \ref{iterimp}, an $( \epsilon, 1/poly(n))$-approximate quantum multistep step (as in Claim \ref{U}) for the iterative method 
can be implemented with cost $O(C_V+\log \tau)$. 
\end{proposition} 
\begin{proof} 
It is easy to see that the complexity of $U$ is asymptotically upper bounded by the complexity of applying the unitary $V^\tau$ on $\ket{0}\ket{r_0}\ket{0}$. For this we first apply $V$ once to get the affine transformation (with running time $\widetilde{O}( \mu(A|b) /\epsilon)$), then we apply the SVE procedure on $A$ to obtain estimates $\overline{\lambda_i}$ of the eigenvalues of $S=(I-\alpha A)$ and then compute $\overline{\lambda_{i}}^{\tau-1}$ (in time $O(\log \tau)$) as the estimates for the singular values of $S^{\tau-1}$ for the conditional rotation. This gives us the unitary $U$ in Claim \ref{U} with error $t\epsilon$ for the $t$-th step. The running time of the second part is $\widetilde{O}( \mu(A) /\epsilon+\log \tau)$. 
As $\mu$ is monotone over the addition of columns, the overall running time is $O(C_V+\log \tau)$. 

\end{proof}

\section{Quantum iterative algorithms} 
\subsection{Linear systems} \label{ls} 
Let $A \succeq 0$ be a $n\times n$ positive semidefinite matrix with eigenvalues in the interval $[1/\kappa, 1]$ and let $b \in \R^{n}$ with $\norm{b}=1$. We assume that $b$ is stored in a QRAM data structure. The goal is to solve the linear system $A\theta=b$. 

The classical iterative method starts with the observation that the quadratic form $F(\theta)=\theta^{T}  A \theta - b \theta$ is minimized at the solution to $A \theta=b$. The algorithm starts with an arbitrary $\theta_{0}$ and applies the following updates,   
\all{ 
\theta_{t+1} = \theta_{t} + \alpha ( b - A\theta_{t}) = \theta_{t} + \alpha r_{t} 
}{plus} 
where the step size $\alpha$ will be a small constant and the residuals $r_{t}:= b- A\theta_{t}$ for $t\geq 0$. 
The residuals satisfy the recurrence $r_{t+1} = b - A(\theta_{t} + \alpha r_{t}) = (I- \alpha A) r_{t}$ and the initial condition $r_{0} =(b - A\theta_{0})$.

The convergence analysis and the choice of the step size $\alpha$ follow from the following classical argument. Let $\theta^{*}=A^{-1}b$ be the optimal solution, the error $e_{t}:= \theta_{t} - \theta^{*}$ satisfies the recurrence $e_{t+1} = (\theta_{t+1} - \theta^{*}) = \theta_{t} - \theta^{*} + \alpha (b - A\theta_{t})
= e_{t} + \alpha A(\theta^{*} - \theta_{t})= (I- \alpha A) e_{t}$. After $t$ steps of the iterative method, 
\[ 
\norm{e_{t}}= \norm{ (I- \alpha A)^{t}   e_{0} } \leq (1- \alpha/\kappa)^{t}  \norm{e_{0}}.
\] 
The method therefore converges to the optimal solution $\theta^{*}$ within error $\epsilon$ in $\tau = O( \kappa \log (\norm{e_{0}}/\epsilon) /\alpha)$
iterations. The step size $\alpha$ can be fixed to be a small constant say $\alpha = 0.01$ and the starting point $\theta_{0}$ chosen to be a unit vector so that 
$\norm{e_{0}} \leq \kappa$. With these choices the number of iterations required for convergence within error $\epsilon$ is $O(\kappa \log (\kappa/\epsilon))$. 

The classical iterative method can be viewed as gradient descent with affine updates. We obtain a quantum iterative method for this problem using Proposition \ref{iterimp}. 

\begin{theorem} \label{qlinsys} 
Given positive semidefinite $A \in \R^{n\times n}, b \in \R^{n}$ stored in the data structure of Theorem \ref{simps}, there is an iterative quantum algorithm that outputs a state $\ket{z}$ such that $\norm{\ket{z} - \ket{A^{-1}b}} \leq \delta$ with expected running time $O(\frac{\kappa(A)^3 \log^3 \frac{\kappa(A)}{\delta} \mu(A | b)}{\delta})$.
\end{theorem}   
\begin{proof} 
An $(\epsilon, 1/poly(n))$ approximate iterative step for the classical iterative method described above can be implemented in time $C_{V}=\widetilde{O}( \mu(A | b) /\epsilon)$ using Proposition \ref{iterimp}. The cost $C_{U}$ of the unitary in Theorem \ref{aim} is the cost of implementing the powers $V^t$ by Proposition \ref{48}, in this case this cost is the same as $C_{V}$. This is because we do not have to apply $V$ sequentially $t$ times, but once the SVE has estimated the eigenvalues, we can directly perform the conditional rotations by an angle proportional to the $t$-th power of each eigenvalue.

Theorem \ref{aim} shows that given an $(\epsilon, 1/poly(n))$ approximate step and for constant $\alpha, \norm{\tilde{\theta}_\tau}$, the 
quantum iterative method has error $\norm{\ket{\theta_\tau} - \ket{\tilde{\theta}_\tau}} = O(\tau^2 \epsilon)$ and requires time $O(\tau C_{U})= O(\tau C_V) = O(\frac{ \tau \mu(A|b)}{\epsilon} )$. We take $\epsilon = O(\frac{\delta}{\tau^2})$ in order to have $\norm{\ket{\theta_\tau} - \ket{\tilde{\theta}_\tau}} \leq \delta/2$  for some $\delta>0$. The running time bound follows as we showed that $\norm{\ket{\theta_\tau} - \ket{A^{-1}b} }  \leq \delta/2$ 
for $\tau = \kappa(A) \log \frac{\kappa(A)}{\delta}$. 

In order to complete the proof, it remains to show that $\norm{\theta_{\tau}}=\Omega(1)$ for $\tau = O(\kappa \log (\kappa/\epsilon))$. The solution $A^{-1}b$ to the linear system has norm at least $1$ as $b$ is a unit vector and the eigenvalues of $A^{-1}$ are greater than $1$. After $\tau$ steps we have $\norm{ \theta_{\tau}  - \theta^{*} } \leq \epsilon \Rightarrow \norm{\theta_{\tau} }   \geq  \norm{\theta^{*} } -\epsilon \geq 1 - \epsilon$.

\end{proof}

\subsection{Weighted Least Squares}\label{wls}
For the weighted least squares problem, we are given a matrix $X \in \R^{m\times n}$ and a vector $y\in \R^{m}$, as well as a vector $w \in \R^{m}$ of weights,  and the goal is to find $\theta \in \R^{n}$ that minimizes the 
squared loss $\sum_{i \in [m]} w_{i} (y_{i} - x_{i}^{T}  \theta)^{2}$. The closed form solution is given by, 
\[
\theta = (X^{T}  W X)^{-1} X^{T}  W y 
\]  
and thus the problem can also be solved using a direct method. However, the direct method is prohibitive for large datasets and iterative methods are used instead as they do not involve solving large linear systems. The iterative method for weighted least squares is a gradient descent algorithm with the update rule $\theta_{t+1}= \theta_{t} + \rho \sum_{ i \in [m]} w_{i} ( y_{i} - \theta_{t}^{T}  x_{i}) x_{i}$ which in matrix form can be written as,  
\all{ 
\theta_{t+1} = (I -  \rho X^{T} WX) \theta_{t}  + \rho X^{T}  W y 
} {owls} 
The quantum iterative method for the weighted least squares problem is given by the following theorem. 
\begin{theorem} \label{qwlsq} 
Let $X \in \R^{m \times n}, y \in \R^{n}, w \in \R^{m}$ be stored in the data structures in Theorem \ref{simps}. Let $W=diag(w)$, $A= X^{T} WX$ and $b= X^{T}  W y$, then there is an iterative quantum algorithm that outputs a state $\ket{z}$ such that $\norm{\ket{z} - \ket{A^{-1}b}} \leq \delta$ with expected running time $O(\frac{\kappa(A)^3 \log^3 \frac{\kappa(A)}{\delta} \mu( \sqrt{W} X | y)}{\delta})$.
\end{theorem}   
\begin{proof} 
The iterative update rule for weighted least squares in equation \eqref{owls} can be written as $\theta_{t+1} = \theta_{t} + \rho r_{t}$ where $r_{t} =  X^{T}  W y  -  X^{T} WX \theta_{t}$. These updates are analogous to the linear system updates in equation \eqref{plus} as $r_{t}=b - A\theta_{t}$ for $b= X^{T}  W y$ and $A= X^{T} WX$. The step size $\rho$ is 
analogous to the step size $\alpha$ for the linear system. It follows from Theorem \ref{qlinsys} that given an $(\epsilon, 1/poly(n))$ approximate step for the iterative method, the quantum iterative algorithm for weighted least squares would have the stated running time. 

 It remains to show how to implement the iterative step for the least squares problem, which is somewhat different from the case of linear systems as 
instead of the matrix $X^{T} WX$ and the vector $X^{T}  W y$, we have the matrix $X$ and vector $y$ stored in memory and weights $w$ arrive as a stream of entries 
$(i, w_{i})$ for $i\in [m]$. Nevertheless, the iterative step can be implemented in this setting as we show next. 

Note that $A = B^{T}  B$ where $B=\sqrt{W} X$, thus the eigenvalues of $A$ are the squared singular values of $B$. It therefore suffices to perform the generalized SVE for $B$. 
We assume that the data structures for performing generalized SVE for $X$ have been created. 

In order to perform the generalized SVE for $B$, we append to the data structure in Theorem \ref{dsplus} a variable $M_{w}$ that is initially a copy of $M$. Recall that the variable $M$ stores the 
maximum row norm $M = \max_{i\in [m]}\norm{a_{i}}^{2}$ for the matrix A stored in the data structure. The variable $M_{w}$ 
is updated whenever $w_{i}$ arrives or $\norm{x_i}$ gets updated so as to store $\max_{i} w_{i}\norm{x_{i}}^{2} $, 
that is if $M_{w} \leq w_{i} \norm{x_{i}}^{2}$ then set $M_{w}=  w_{i} \norm{x_{i}}^{2}$. We replace $M \to M_{w}$ and $\norm{x_{i}} \to \sqrt{w_{i}} \norm{x_{i}}$ in equation \eqref{tag} and follow exactly the same steps 
as in Theorem \ref{dsplus} to implement the unitary, 
\als{ 
&U' \ket{  i, 0^{\lceil \log (n+1) \rceil} } = \ket{i} \frac{1}{\sqrt{M_w}} \big ( \sum_{j \in [n]} \sqrt{w_{i}} x_{ij} \ket{j} + \nl 
&(M_{w} - w_{i} \norm{x_{i}}^{2})^{1/2}\ket{n+1}  \big )  
} 
Using $U', V$ instead of $U, V$ in Theorem \ref{sveplus1} we can perform generalized SVE for $B=\sqrt{W} X$
in time $\mu(B)$. In order to multiply by $A=B^{T}B$ for the iterative method, we perform generalized SVE for $B$ in equation \eqref{istep1} and 
then conditional rotation with factor $\alpha \overline{\sigma_{i}}^{2}$.   

Analogous to the above procedure one can also implement matrix multiplication for $B' = X^{T} W$. 
Note that the state $\ket{b}$ is not in the memory, so we can not do the first affine update used for linear systems (where we set $r_0 = b-A\theta_0$ for a random $\theta_0$). Instead we have $y$ and $X$ in memory, so we first do the affine step to create $(y-X\theta_0)$ and then multiply with the matrix $X^TW$, 
\[
b-A\theta_0 = X^TWy-X^TWX\theta_0 = X^TW (y-X\theta_0).
\] 
\end{proof} 
The procedure for performing generalized SVE for $B=\sqrt{W} X$ in the proof of Theorem \ref{qwlsq} implies that given the data structure for storing $X \in \R^{m\times n}$
and a weight vector $w\in \R^{m}$, it is possible to prepare the data structure for performing SVE for $B=\sqrt{W} X$ in time $O(m)$ where $W=diag(w)$. We note this fact as it may be useful 
for other applications. 

It is straightforward to add $\ell_{2}$ regularization to the weighted least squares problem. The loss function becomes $\sum_{i} w_{i} ( y_{i} - \theta^{T}  x_{i})^{2} + 
\lambda \norm{\theta}^{2}$ and the update rule changes to $r_{t} = b - A\theta_{t}$ for $b=X^{T} W y$ and $A= X^{T}  WX+ \lambda I$. The algorithm 
performs the generalized SVE for $X^{T}  WX+ \lambda I$ instead of $X^{T}  WX$.

\subsection{Stochastic gradient descent for Weighted Least Squares} 
In the classical setting, it is expensive to compute 
the gradient $\sum_{ i \in [m]} w_{i} ( y_{i} - \theta_{t}^{T}  x_{i}) x_{i} $ by summing over the entire dataset 
when the dataset size is large. Moreover, due to redundancy in the dataset, the gradient can be estimated by summing over randomly sampled batches. Stochastic gradient descent utilizes this fact in the classical setting to lower the cost of the updates. 
Stochastic gradient descent algorithms do not compute the gradient exactly, but estimate it over 
batches $\sum_{ i \in S_{j}} w_{i} ( y_{i} - \theta_{t}^{T}  x_{i}) x_{i}$ obtained by randomly partitioning the dataset.

In the quantum case, applying a linear system solver or the iterative method in Theorem \ref{qwlsq} for a 
large dataset would require a large-sized QRAM and coherent operations over a large number of qubits. 
Stochastic gradient descent remains relevant for quantum iterative methods since it can considerably reduce the size of the QRAM as well as the number of qubits on which we 
need to perform coherent operations. 

The stochastic gradient updates are defined for any choice of partition $\Pi= (S_{1}, S_{2}, \cdots, S_{k})$ for the dataset where $S_{i} \subset [m]$ for $i\in [k]$ are subsets of 
$[m]$ of equal size. The random partitioning can be easily implemented in the quantum setting by permuting the data before entering into the QRAM, thus storing matrices $A_{j}|b_{j}$ for batches of a fixed size. For a given partition $\Pi$ 
let $X_{j}$ be the the $|S_{j}| \times n$ matrix obtained by selecting the rows corresponding to $S_{j}$. Define $A_{j} = X_{j}^{T}  W_{j} X_{j}$ 
where $W_{j}$ is the diagonal matrix of weights restricted to $S_{j}$. 

The stochastic gradient descent algorithm starts with the initial condition $r_{0} =(b - A_{1}\theta_{0})$ and iteratively applies the updates $r_{t} = (I- \rho A_{t'}) r_{t-1}$ on the residuals with $t' = t +1 \mod k$. Note that as the number of steps $\tau$ is larger than the number of partitions $k$, we cyclically iterate the updates corresponding to the matrices $A_{j},  j \in [k]$. It is straightforward to implement these updates efficiently using Theorem \ref{qwlsq} as the matrices $A_{j}$ and the weights $W_{j}$ are stored in memory. The running time is linear in the parameter $\mu = \max_{i \in [k]} \mu(X_{k})$. 

The updates in the classical stochastic gradient descent algorithm are affine, our quantum iterative method can therefore simulate these updates. The convergence analysis in equation \eqref{three} is not applicable
for stochastic gradient descent, we refer to \cite{BCN16} for the classical convergence analysis. The correctness follows as the quantum algorithm follows is able to simulate each step of the classical stochastic gradient algorithm 
with sufficient precision. 

In the case of linear systems and least squares, the updates were of the form $r_{t+1} = (I- \rho A) r_{t}$ for a fixed matrix $A$, and hence we could simultaneously apply $t$ steps of the update in one step and hence have running time $C_{U}=O(C_V)$ using Proposition \ref{48}. In the stochastic gradient descent case, we have $k$ different matrices $A_{t}$ where 
$k$ is the number of partitions for the dataset, and these matrices have different eigenbases. Therefore, we can only perform the linear updates sequentially, and the cost $C_{U}$ for implementing $U$ is 
$O(\tau C_{V})$ in this case. Claim \ref{aim} therefore implies that the running time for the stochastic gradient descent is $C_{U}=O(\tau^{2} C_{V})$ as opposed to $C_{U}=O(\tau C_{V})$ for the linear systems and weighted least squares. 

Similar to the running time analysis in Section \ref{itertime} we can also compute the running time for quantum stochastic gradient descent
algorithm. If we use our linear system solver with precision dependence $O(1/\epsilon)$, then the running time is $O(\tau^{4} \mu)$. However if we use a
quantum linear system solver with precision dependence $\log (1/\epsilon)$, then the running time becomes $\widetilde{O}(\tau^{2} \mu)$.

\bibliographystyle{IEEEtranS} 
\bibliography{bibliography.bib}

\end{document}